\documentclass[letterpaper,twocolumn,10pt]{article}
\usepackage{usenix}

\usepackage{tikz}
\usepackage{amsmath}

\usepackage{graphicx} % Required for inserting images
\usepackage{cite}

\usepackage{mathtools}
\usepackage{amsmath}
\usepackage{amssymb}
\usepackage{amsthm}
\usepackage{thm-restate}
\usepackage{algorithm, algpseudocode}
\usepackage{xspace}
\usepackage{xurl}
\usepackage[T1]{fontenc}
\usepackage{enumitem}

\newtheorem{theorem}{Theorem}

\usepackage{subcaption}
\usepackage{xcolor}

\DeclarePairedDelimiter\ceil{\lceil}{\rceil}

\usepackage{cleveref}

\newtheorem{definition}{Definition}   % independent counter
\newtheorem{lemma}{Lemma}             % independent counter

\title{\Large \bf \sys \\ Efficient, and Censorship-Resilient Signature Aggregation for Million Scale Consensus}

\begin{document}

\date{}

\newcommand{\sys}{\textsc{Wonderboom}\xspace}

\newcommand{\ray}[1]{{#1}}
\newcommand{\mxy}[1]{{#1}}
\newcommand{\kp}[1]{{#1}}
\newcommand{\za}[1]{{#1}}

\algtext*{EndIf}% Remove "end if" text
\algtext*{EndFor}% Remove "end if" text
\algtext*{EndProcedure}% Remove "end if" text
\algtext*{EndFunction}% Remove "end if" text

%for single author (just remove % characters)
\author{
{\rm Anonymous}\\
Submission
%\and
%{\rm Second Name}\\
%Second Institution
% copy the following lines to add more authors
% \and
% {\rm Name}\\
%Name Institution
} % end author

\author{Zeta Avarikioti\inst{1} \and
Ray Neiheiser \inst{2} \and
Krzysztof Pietrzak \inst{2} \and
Michelle X. Yeo \inst{3}
}

\author{
{\rm Zeta Avarikioti}\\
TU Wien \& Common Prefix
\and
{\rm Ray Neiheiser}\\
Institute of Science and Technology Austria
\and
{\rm Krzysztof Pietrzak}\\
Institute of Science and Technology Austria
\and
{\rm Michelle X. Yeo}\\
Nanyang Technological University and Aarhus University
} % end author

%\institute{TU Wien \& Common Prefix \and 
%Institute of Science and Technology Austria \and
%Nanyang Technological University and Aarhus University
%}
%Springer Heidelberg, Tiergartenstr. 17, 69121 Heidelberg, Germany
\maketitle              % typeset the header of the contribution

\begin{abstract}

% Ethereum = 349b + 329b = 678 > Singapore, Mexico and Thailand
% top 25

Over the last years, Ethereum has evolved into a public platform that safeguards the savings of hundreds of millions of people and secures more than \$650 billion in assets, placing it among the top 25 stock exchanges worldwide in market capitalization, ahead of Singapore, Mexico, and Thailand.
As such, the performance and security of the Ethereum blockchain are not only of theoretical interest, but also carry significant global economic implications.

At the time of writing, the Ethereum platform is collectively secured by almost one million validators highlighting its decentralized nature and underlining its economic security guarantees. 
However, due to this large validator set, the protocol takes around 15 minutes to finalize a block which is prohibitively slow for many real world applications. This delay is largely driven by the cost of aggregating and disseminating signatures across a validator set of this scale. Furthermore, as we show in this paper, the existing protocol that is used to aggregate and disseminate the signatures has several shortcomings that can be exploited by adversaries \mxy{to shift stake proportion from honest to adversarial nodes}.

In this paper, we introduce \sys, the first million scale aggregation protocol that can efficiently aggregate the signatures of millions of validators in a single Ethereum slot (x32 faster) while offering higher security guarantees than the state of the art protocol used in Ethereum.

Furthermore, to evaluate \sys, we implement the first simulation tool that can simulate such a protocol on the million scale and show that even in the worst case \sys can aggregate and verify more than 2 million signatures within a single Ethereum slot.

\end{abstract}

\section{Introduction}
\label{sec:introduction}

After the inception of Bitcoin in 2014~\cite{bitcoin}, the term blockchain was coined to describe systems built on similar architectural and conceptual foundations. Since then, hundreds of blockchain platforms emerged, each modifying the original design to improve throughput~\cite{DBLP:conf/podc/KeidarKNS21}, reduce ecological impact~\cite{CohenP23}, or support new use cases. 

The introduction of smart contracts extended blockchain functionality by enabling arbitrary programs to run on-chain. This made it possible to create and deploy programmable assets within blockchain ecosystems, allowing users to trade, manage, and interact with these assets directly on the network.

Among these platforms, Ethereum is the largest programmable asset blockchain ecosystem, holding over \$650 billion in assets, which would place it in the top 25 stock exchanges world-wide~\cite{largeststock,ethereumvalue}. Furthermore, due to the introduction of stablecoins (i.e., tokens that aim to mirror the real world value of fiat tokens such as the USD), millions of people, especially from high inflation countries, rely on Ethereum to store their savings in foreign currencies~\cite{statista_crypto2025}.

At the time of writing, Ethereum is secured by the largest validator set among all smart contract platforms, with almost one million active validators, reinforcing its decentralized and permissionless nature. The Ethereum protocol builds on early research in Byzantine Fault Tolerant (BFT) consensus. Thus, for a transaction to be considered final, Ethereum requires that, out of $N$ validators, at least two consecutive quorums of $\frac{2N}{3}$ attestations have to be collected~\cite{pbft}.

However, due to its large validator set, Ethereum suffers from two major limitations.
First, finality is slow: on average, clients must wait 15 minutes until their transaction is considered final~\cite{beaconchain,singleslotpath}, making it unrealistic for usage in everyday payment scenarios.
Second, as validator rewards in Ethereum depend on the timely inclusion of their attestation, the protocol is vulnerable to censorship attacks during the quorum collection process. Even without compromising safety, adversaries can delay or censor honest attestations frequently enough to cause large aggregate reward losses (on the order of \$100M+ per year, as we show in Section~\ref{sec:eth-background}), gradually shifting stake distribution from honest to adversarial over time.

The primary reason behind this large finality latency comes from the high computational and bandwidth cost of aggregating and verifying $\ge \frac{2N}{3}$ signatures per block across millions of validators.
To make this feasible without significantly increasing hardware requirements for all validators, Ethereum partitions the validator set into smaller subsets of $32$ committees of $\frac{N}{32}$ validators each. In each slot, all signatures from a committee are aggregated. Since even $\frac{N}{32}$ signatures represent a substantial number at Ethereum's scale, a hierarchical tree aggregation protocol is used: each committee is further divided into up to $64$ sub-committees of at least $128$ validators, each aggregating up to $2048$ signatures, which are then gossiped to the next proposer. As a result, it takes over $32$ slots (i.e., one epoch) to verify and aggregate all $N$ validator signatures and two consecutive epochs to achieve finality~\cite{pbft,singleslotpath}.

The most effective way to balance this aggregation load is to use deeper aggregation trees~\cite{kauri}. However, this approach faces two major challenges.
First, Ethereum relies on a gossip network for message propagation between validators to protect proposers from censorship attacks, which, e.g., could be motivated by MEV opportunities.
However, disseminating messages through the gossip network takes a significant amount of time (e.g., $\approx 4$ seconds in Ethereum). 
Thus, for any given slot time, the inherent latency of the gossip protocol limits the potential depth of the aggregation tree.
For example, in the case of Ethereum with 12-second slots, we are limited to a tree of depth 2 with 3 rounds of gossip.
Second, deeper aggregation trees create more opportunities to censor validators, which may cause significant reward losses or penalties for the affected validators and may also delay finality.

This raises the primary research question of our work: 

\noindent\emph{Can we design an aggregation protocol that preserves Ethereum's decentralization (millions of validators) while achieving fast finality and strong, quantifiable guarantees against censorship?}

In this work, we answer this question in the affirmative and present \sys, a tree-based aggregation protocol that aggregates \emph{all $N$ signatures every slot}, thereby achieving two-slot finality with explicit inclusion guarantees at scale.
We model \sys as a tree with a fanout of $m$, where each node represents a committee of validators. As in Ethereum, each committee comprises 128 members, with 16 random representatives.

To quantify and show the censorship-resilience of \sys, we distinguish between two types of censorship attacks. \emph{Proposer censorship} occurs when a proposer is prevented from proposing a block, resulting in the loss of block rewards and penalties for missing the slot. \emph{Vote censorship} occurs when validators are prevented from including their vote in the next block, leading to lost rewards and penalties.

To mitigate proposer censorship, in \sys the proposer and validator processes are physically separated and have distinct IP addresses~\footnote{The validator process can be implemented as a gateway, analogous to the gateway nodes commonly used by mining pools in PoW Ethereum~\cite{gateways}.}. This separation reduces the attack surface for proposer deanonymization via the gossip layer, as studied in~\cite{heimbach2024deanonymizing} and allows us to leverage point-to-point communication channels between the validator processes in the aggregation protocol while preserving the anonymity of the proposer by retaining gossip-based communication for messages to and from block proposers.

%To combat proposer censorship, in \sys we propose physically separating the proposer from the validator process such that they have distinct IP addresses. This makes it significantly harder to deanonymize proposers on the gossip layer as in~\cite{heimbach2024deanonymizing}, as it would result in much less frequent interactions with the gossip layer which significantly reduces the success rate of this attack. Furthermore, this allows us to to eliminate the need for the gossip protocol in the aggregation protocol, while still using gossip for messages to and from block proposers.

To address vote censorship attacks, \sys adds a simple forwarding rule where the largest aggregate is combined with one uniformly random aggregate. We pair this with a once-per-$k$ reward allocation where a validator earns the reward for a $k$-slot window if it got its vote included at least once within that window. We show that this simple rule allows us to lower-bound the per-slot inclusion probability $p$ and shows that the chance of missing out on rewards is $(1-p)^k$, which declines exponentially with $k$. These choices neutralize long-term stake drift from vote censorship while preserving Ethereum's decentralization. 

We also present an analysis comparing the censorship-resilience of \sys with Ethereum and show that \sys achieves censorship resilience with a smaller $k$ parameter compared to Ethereum: 
for $k = 64$ (two epochs) in \sys, our analysis indicates that under the same adversarial assumptions and current Ethereum reward allocation rules, Ethereum would need to extend its effective penalization window to six epochs to reach a similar level of resilience.
% for $k = 64$ (two epochs) in \sys, Ethereum would need to extend its penalization window to six epochs to reach a similar level of resilience.
% Furthermore, we show that the current analysis of vote censorship in Ethereum relies on an inaccurate metric, which falsely inflates the vote censorship resilience guarantees by several orders of magnitude.
Furthermore, we show that the current analysis of vote censorship in Ethereum relies on an imprecise metric, which significantly overstates the vote censorship resilience guarantees by several orders of magnitude.
% Furthermore, we show that the standard metric used to analyze vote censorship in Ethereum fails to capture representative-level inclusion effects, thereby overstating vote-censorship resilience by several orders of magnitude.

We substantiate our theoretical analysis by implementing the first simulator for million scale consensus that can emulate the worst-case complexity (e.g., $\frac{N}{3}$ faulty nodes) and show that \sys can aggregate more than $2$ million validator signatures within a single Ethereum slot, enabling two-slot finality. Beyond performance, we formalize vote-censorship resilience at scale and obtain stronger inclusion guarantees for \sys compared to Ethereum's current design.

\smallskip\noindent{\bf Summary of Contributions:}
\begin{itemize}
    \item We present \sys, a censorship-resilient aggregation protocol that can aggregate millions of signatures in a single slot. Notably, due to its modular design, \sys can be integrated into existing large-scale systems such as Ethereum with minimal changes.
    \item We present a theoretical analysis of the vote-censorship resilience of Ethereum's current aggregation protocol, and show that \sys achieves significantly better vote-censorship resilience compared to Ethereum.
    \item We provide the first simulation framework capable of evaluating the worst-case performance of large-scale Byzantine fault tolerant signature aggregation, addressing the challenges of assessing systems at the million-validator scale such as Ethereum.
    \item We evaluated \sys using our simulation tool and demonstrated that it can aggregate $N$ signatures over \textbf{x4} faster than the Ethereum design, while providing superior censorship resilience at scale. In combination with the proposer-validator separation, this results in a total speed-up of over \textbf{x32}.
\end{itemize}

\smallskip\noindent{\bf Organization.}
The remainder of the paper is organized as follows. In~\Cref{sec:systemmodel} we model Ethereum's block building process, present our desideratum, as well as state our model assumptions. 
We then go on to describe the current signature aggregation protocol in Ethereum in~\Cref{sec:eth-background}, as well as outline the limitations and vulnerabilities of the current Ethereum aggregation protocol.
~\Cref{sec:wonderboom} details \sys, our tree-based signature aggregation protocol that can aggregate all validator signatures in a single slot.
We analyze the vote-censorship resilience of \sys against a slowly-adaptive adversary in~\Cref{sec:analysis} and show that it achieves better vote-censorship resilience than Ethereum.
In~\Cref{sec:evaluation} we present our simulation tool and the evaluation of the worst-case performance of \sys when scaling up to millions of nodes.
We then discuss further optimizations and connections of our protocol in~\Cref{sec:discussion}, discuss related work in~\Cref{sec:relatedwork}, and conclude in~\Cref{sec:conclusion}.
\section{Model}
\label{sec:systemmodel}
For $n \in \mathbb{N}$, we use $[n]$ to denote $\{1, \dots, n\}$.
The system consists of $N$ server processes $p_1, \ldots, p_N$ that are connected via perfect point-to-point channels, implemented using mechanisms for message reordering, deduplication, and retransmission. 
As we aim to provide solutions that are compatible with Ethereum where validators are penalized if they do not participate timely, we have to assume that we operate in the synchronous model. That is, we assume some globally-known and finite bound $\Delta\geq0$ such that any message sent by a process at time $t$ is guaranteed to be delivered to all other processes by time $t+\Delta$.

\ray{Note that this synchrony bound is strictly necessary for systems where slashing is used as an accountability measure~\cite{10.1145/3670865.3673548}. However, as observed in~\cite{10.1145/3670865.3673548}, the protocol may still make progress under a smaller $\delta < \Delta$.}

%As such, when a process $p_i$ sends a sequence of messages $b_i, b_{i+1}$ to a process $p_j$, $p_j$ eventually delivers $b_i$ and $b_{i+1}$ in the intended order.
We model block building as proceeding in slots. In each slot, a leader process proposes a block and then has to collect a sufficient amount of votes such that the block is considered valid and can thus be appended to the blockchain. We model a set of $r$ consecutive slots as an epoch (e.g., $r=32$ for an Ethereum epoch).

We model \sys as a black box that provides vote aggregation for any leader-based consensus protocol~$\Omega$, where in~$\Omega$ the leader invokes \sys to obtain the aggregate votes for a block~$B_i$ in slot~$i$. As long as the consensus protocol can reliably verify that a sufficient fraction of nodes contributed signatures to the aggregate (e.g., guaranteed by constructions as BLS signatures~\cite{bls,bls2}), \sys does not interfere with safety. Consequently, in \sys the main concern is ensuring liveness, in particular by preventing vote-censorship attacks that may suppress the votes of correct nodes. 

More formally, the desired property \sys aims to guarantee is \emph{vote-censorship resilience}, which is defined as follows:

\begin{definition}[Vote-censorship resilience]\label{def:er}

   A protocol is $\alpha$-vote-censorship-resilient for $\alpha=1-(1-p)^k$ based on the probability $p$ that a validator is censored in a single round and the probability $(1-p)^k$ over $k$ consecutive rounds, which declines exponentially in $k$.
    
    %Let $p$ denote the probability that an honest node's vote is not censored in a single slot.
   % A protocol $\Pi$ is $(p,k)$-vote-censorship resilient if there is some parameter $k$ such that the probability that no honest nodes get censored over $k'>k$ consecutive slots converges to~$1$.
\end{definition}

%As such, we are interested in protocols with small values for $p$ such that we can pick small $k$ and still guarantee that honest validators do not miss out on rewards due to vote-censorship attacks while also guaranteeing with small $k$ that validators participate consistently in the protocol.

As such, we aim to design a protocol with a small per-slot censorship probability $p$, such that validators are already protected from vote-censorship attacks for small values of $k$, while also reducing the potential for free-riding, i.e., nodes participating only every $k$ rounds due to lack of incentive to participate each round. We discuss free-riding in more detail in Section~\ref{sec:discussion}.

Furthermore, in~\Cref{sec:analysis}, we show that \sys allows choosing a $k$ that is three times smaller than in Ethereum, indicating that \sys is significantly more resilient to vote-censorship attacks.

%As such, we focus on protocols with low per-slot censorship probability $p$, enabling a choice of $k$ protecting honest validators from missing out on rewards due to vote-censorship attacks while limiting 

%In particular, we are interested in protocols that ensure $(p,k)$-vote-censorship resilience with small values of $k$. This implies stronger resilience to adversarial censorship attempts. In~\Cref{sec:analysis}, we show that the $k$ parameter in Ethereum is three times higher compared to \sys, implying that \sys is more resilient to vote-censorship attacks. 

%\sys operates in the partially synchronous model, where message delays may be unbounded during periods of asynchrony. However, after some unknown Global Stabilization Time (GST), messages are delivered within a known bound~$\Delta$ for long enough to enable progress. Thus, according to the FLP impossibility result~\cite{flp}, \sys can also only guarantee liveness after GST.
%\ray{Delete partial synchronous}

\smallskip\noindent{\bf Adversarial model and cryptographic assumptions.}
We assume the Byzantine fault model, where up to $f$ out of $N = 3f + 1$ nodes can be faulty and may deviate arbitrarily from the protocol. Byzantine nodes may behave arbitrarily and collude, but they are assumed to be computationally bounded, i.e., standard cryptographic primitives hold. 
%That is, they run in probabilistic polynomial time (PPT).
Furthermore, we consider a \emph{slowly adaptive} adversary: the adversary may update the set of corrupted nodes only at the beginning of each epoch, but the corruption set remains fixed throughout the epoch.  

We assume that the blockchain provides a Public Key Infrastructure (PKI) functionality, i.e., each process has a unique public-private key pair, allowing it to sign messages with its private key. The signatures can be publicly verified in the system using the corresponding public key. A validator $i$ can obtain the public key of another validator $j$ by querying the PKI for $pki[j]$.

We use Boneh-Lynn-Shacham (BLS) signatures~\cite{bls} to enable the aggregation of $N$ signatures and $N$ public keys into a single signature–public key pair that can be verified in the same time as a single signature.  
To construct the compressed public key in a distributed fashion, we require a bitvector structure that identifies which processes contributed their signature to the aggregate. This requires a system that assigns each public key a unique index in the bitvector.

%As the bitvector may become large for a high number of participants (e.g., 1 megabit for 1 million participants), we refer to compression structures such as Roaring Bitmaps~\cite{roaring} to reduce its size as discussed in Section~\ref{sec:discussion}.

\smallskip\noindent{\bf Glossary:}
Here we define terms that we use throughout our paper.

\begin{description}
\item \textbf{Validator}: A participant in the Ethereum consensus protocol responsible for attesting to (voting on) the validity of proposed blocks.
\item \textbf{Proposer}: A validator elected to build and broadcast a new block in a specific slot.
\item \textbf{Slot}: A fixed period of time (e.g., 12 seconds) during which a single block is expected to be proposed and attested.
\item \textbf{Epoch}: A period of a fixed number of consecutive slots (e.g., 32 slots), which aligns with how the protocol processes validator-set changes.
\item \textbf{Committee}: A group of validators.
\item \textbf{Vote Aggregation}: The process of combining individual validator attestations into a compact aggregate signature.
\end{description}

\section{Ethereum: Design \& Limitations}
\label{sec:eth-background}

\subsection{Background}\label{subsec:eth-mechanics}

Ethereum~\cite{ethereum} is a Proof of Stake (PoS) blockchain where participants must lock 32 ETH (over 96,000 USD at the time of writing) to become a validator. There are currently almost 1 million validators participating in the Ethereum protocol.

Ethereum proceeds in epochs of 32 slots of 12 seconds each, where in each slot the signatures of a committee of roughly $\frac{N}{32}$ validators are verified and aggregated. Each committee is further divided into 64 subcommittees of at least 128 and at most 2048 validators. Validators broadcast their attestations within their respective subcommittee, and on expectation, 16 randomly selected \emph{committee representatives} aggregate the signatures. These per-subcommittee aggregates are then gossiped to the proposer, who includes the best (largest-weight) aggregate from each committee in the block.
Thus, it takes 32 slots (one epoch) for all validators to contribute their signatures. Finality therefore requires two consecutive epochs with at least $\frac{2}{3}$ participation, resulting in an average finality latency of about 15 minutes~\cite{singleslotpath}.

\smallskip\noindent{\bf Incentives and Penalties.}
% \label{subsec:eth-incentives}
Attestation rewards are tied to timely inclusion. A missed or delayed attestation reduces the validator rewards (e.g., a one-slot delay cuts 50\% of the attestation reward), creating direct economic exposure to any mechanism that censors votes. Thus, sustained censorship redistributes the stake over time from honest to adversarial validators, degrading safety margins~\cite{penalties}. 

\smallskip\noindent{\bf Propagation and Latency Constraints.}
% \label{subsec:eth-latency}
A 12-second slot must accommodate block proposal, network dissemination, committee aggregation, and delivery to the next proposer. Measurements indicate that reaching $\approx 98\%$ of validators via gossip alone can take up to 4 seconds; with two gossip phases, this leaves only a few seconds budget for verification and aggregation, complicating single-slot aggregation on the current scale~\cite{ethereum4sresearch}.

\subsection{Vote Censorship Vulnerabilities}
\label{subsec:eth-censorship}

\smallskip\noindent{\bf Committee-driven censorship (malicious committee supermajority).}
Ethereum enforces a minimum committee size of $128$, based on the argument that with this size, the probability of any single committee containing a $2/3$ Byzantine supermajority is negligible~\cite{ethereumbookcommittees}. The probability of drawing $k$ faulty validators without replacement when sampling a committee of size $n$ from a population of size $N$ with $f$ faulty validators is hypergeometric~\cite{feller1}:
\begin{equation}\label{eq:hyper}
    \Pr[X = k] \;=\; \frac{\binom{f}{k}\binom{N-f}{n-k}}{\binom{N}{n}}
\end{equation}
Using~\Cref{eq:hyper}, the probability of any single committee of size $n$ containing at least a $2/3$ Byzantine supermajority is thus $
\Pr[X \ge 2n/3] = \sum_{k=\lceil 2n/3\rceil}^{n} \Pr[X=k]$.
For $n=128$ and $N=1{,}000{,}000$, one obtains $\Pr[X \ge 86] \approx 5.55 \times 10^{-15}$~\cite{ethereumbookcommittees}.

We argue that this is not the appropriate metric as the above bound concerns \emph{committee composition}, not \emph{committee representation}, i.e. the metric that actually matters is the probability that all 16 representatives are faulty, which is what is directly tied to inclusion and censorship. In Ethereum, each of the $128$ committee members is independently selected as a representative with probability $1/8$ (so the expected number of representatives is $16$). Let $A$ be the event ``validator is Byzantine'' with $\Pr[A]=1/3$ and $B$ the event ``validator is selected as a representative'' with $\Pr[B]=1/8$. A simple upper bound that treats selection and corruption at the node level is then given by:

\begin{align*}
&\Pr[\text{all selected representatives are faulty}] \\
&\approx \big(\Pr[A \vee \neg B]\big)^{128}\\
&= \Big(\tfrac{1}{3} + \tfrac{7}{8} - \tfrac{1}{3}\cdot\tfrac{7}{8}\Big)^{128}\\
&\approx 1.45 \times 10^{-5}
\end{align*}
This is several orders of magnitude larger than the hypergeometric tail above and reflects the risk that \emph{no honest representative is available to include honest votes}, enabling representative-driven censorship. Aggregated over $7{,}200$ slots per day and at least $64$ committees per slot, the daily probability that \emph{at least one} committee experiences such censorship is $1 - \big(1 - 1.45\times 10^{-5}\big)^{7{,}200 \cdot 64} \approx 0.9987$.
This corresponds to an expected value of over 6 committees per day. As such, representative-driven censorship events are expected on a daily basis under these parameters. Economic penalties for delayed or missed inclusion~\cite{penalties} make this gap between committee safety and inclusion safety practically consequential.

\smallskip\noindent{\bf Proposer/Leader Censorship and Economic Impact.}
\label{subsubsec:root-censorship}
The proposer/leader is a single point where censorship can have an over-proportional impact. With $f \approx N/3$ Byzantine power, the proposer is Byzantine with probability $\approx 1/3$ in any given slot~\footnote{\ray{Note that in the current Ethereum ecosystem two entities combined already own over 30\% of the total Stake~\cite{beaconchainentities}}}. Even if safety (finality) is unaffected, intermittent malicious proposers can strategically omit aggregates, causing recurring inclusion losses. A back-of-the-envelope estimate illustrates the systemic effect: if, in Ethereum, in a fraction $1/3$ of slots the proposer is malicious and can suppress up to $1/3$ of honest attestations, then honest validators suffer an expected reward loss factor of $\frac{2}{3}\;\cdot\;\frac{1}{3}\;\cdot\;\frac{1}{3} \;=\; \frac{2}{27}$ (i.e., of the $\frac{2}{3}$ fraction of honest nodes, up to $\frac{1}{3}$ are affected every $3$ rounds).
Using a representative per-attestation reward of $\$0.03$ (about $\$2{,}465$ per year under full participation)~\cite{figmentrewards}, this corresponds to roughly \$205 annualized loss per honest validator. With a conservative baseline of $1$M validators and $2/3$ honest stake, the aggregate annual damage exceeds \$136M, arising purely from inclusion failures rather than safety violations.

\smallskip\noindent{\bf Implications.}
The analysis above indicates that censorship risk in Ethereum is governed primarily by \emph{proposer selection} rather than by \emph{committee composition}. Because attestation rewards are tied to timely inclusion, repeated proposer-driven omissions may induce systematic reward loss for honest validators; over time this shifts stake toward the adversary and weakens safety margins. These considerations motivate an explicit $\alpha$-vote-censorship–resilience target (Def.~\ref{def:er}): penalize only when a validator is excluded for $k$ consecutive slots, thereby aligning incentives while bounding liveness loss. We examine the implications of different values of $k$ in Section~\ref{sec:discussion}.

We next present \sys, which achieves this goal while preserving Ethereum's decentralization and economic-safety properties. By aggregating all $N$ validator signatures each slot and employing a tree design with randomized aggregate selection, \sys enables single-slot aggregation and two-slot finality ($\approx 24$ seconds), in contrast to the current $\approx 15$ minutes finality in Ethereum.

\section{\sys}
\label{sec:wonderboom}

In this section, we introduce \sys, a tree-based aggregation protocol that can handle signatures from millions of validators within the duration of a single Ethereum slot, while simultaneously ensuring censorship resilience. We adopt a structure similar to Ethereum by dividing validators into committees of 128 validators, each with 16 random committee representatives. Compared to Ethereum, \sys constructs a deeper tree in which every node in the tree corresponds to a committee of 128 validators.
The difficulty of constructing such a protocol is twofold. First, it requires designing a protocol that is provably resilient against censorship attacks. Second, we need it to be able to identify tree structures that can efficiently aggregate this large number of signatures.
In the following, we detail the structure and guarantees of \sys.

\begin{figure*}[ht!]
    \centering
    \includegraphics[width=0.7\linewidth]{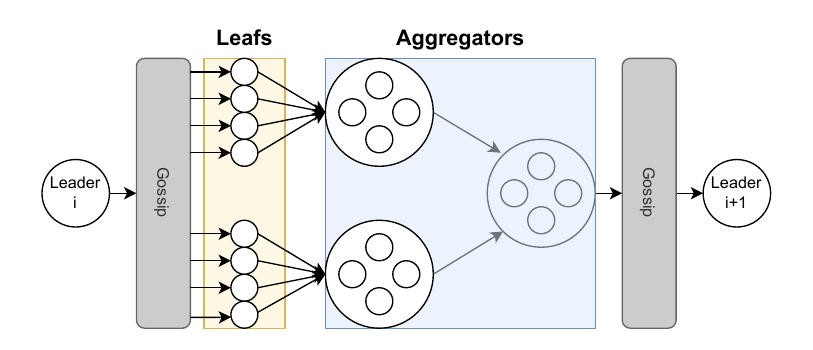}
    \caption{\sys system architecture.}
    \label{fig:architecture}
\end{figure*}

\subsection{Public IPs}
\label{sec:ip}

% The following might also be moved to another section. It doesn't 100% perfectly fit here.

At the time of writing, Ethereum relies on its gossip network to obfuscate the IP addresses of validators, thereby preventing proposers from being censored due to attacks (e.g., DDOS). This is important, as there is an inherent incentive for validators to attack the previous proposer: if the previous proposer fails to disseminate their block within the designated time, the block reward and MEV opportunities are passed to the next proposer. Furthermore, proposers incur penalties when they fail to propose in their designated slot, reducing their stake in the system. As such, a malicious validator can increase their relative stake and reduce the stake of correct validators through successful censorship attacks.

However, as recent research has shown~\cite{heimbach2024deanonymizing}, it is straightforward to deanonymize validators in the existing gossip network, enabling rational and Byzantine validators to execute the aforementioned attacks.
The deanonymization process is facilitated by the regular interaction of validators with the gossip network to broadcast consensus votes (at least once every 6.4 minutes). This allows adversaries to sample gossip messages at different positions in the gossip network to identify proposers.
We therefore propose splitting validators into two distinct processes with distinct IP addresses: a validator process that participates in the aggregation protocol, and a proposer process that builds block proposals.
We discuss in Section~\ref{sec:discussion}, how this can be set up without significantly increasing the hosting costs or complexity for validators.

In order to obfuscate the proposer IP, communication with and from the proposer such as block broadcast and sending the final aggregates to the proposer still has to go through Ethereum's gossip layer.
However, the infrequent interaction with the gossip network makes it significantly harder to identify their IP address, thereby making it also harder to execute proposer censorship attacks.
As a result, communication between the validators can take place through point to point channels, substantially reducing protocol latency.
This allows us to partially bypass the current limitations of the gossip network, where messages can take up to 4 seconds to reach $98\%$ of processes, which made it prohibitive to design more sophisticated aggregation protocols with deeper trees~\cite{ethereum4sresearch}.

\subsection{Architecture}

Figure~\ref{fig:architecture} illustrates the architecture of \sys.
Since communication to and from the proposer must still traverse the gossip layer, there are two phases in \sys that rely on gossip: (i) the initial broadcast of the proposal, and (ii) the transmission of the final aggregates to the next proposer to be included in their block.

Because it can take up to 4 seconds for 98\% of validators to receive a proposal through the gossip layer, given the two gossip phases and the 12 second slot-time in Ethereum, to guarantee that \sys can operate within the existing Ethereum specifications \sys has approximately 4 seconds to verify the previous block, verify and aggregate $N$ signatures and transmit signature aggregates representing all $N$ validators to the next proposer.

There are two main types of validators in the aggregation protocol.
First, \textit{leaf validators} receive a block proposal through the gossip layer, verify the proposal and then send their signature to their respective committee.
Second, each committee consists of $128$ \textit{aggregators} where, analogously to Ethereum, $16$ committee representatives are randomly selected to aggregate the signatures and forward their aggregate to the next layer of aggregators or, through the gossip layer to the next proposer.

Aggregators can be further divided into \textit{leaf aggregators}  that receive individual signatures from up to $m$ (fanout of the tree) leaf validators. And \textit{internal aggregators} that receive $16$ signature aggregates from $\frac{m}{16}$ individual committees. This results in a constant fanout of $m$ at any level of the tree.

%Given at most $f \leq \frac{N-1}{3}$ faulty nodes, the probability that a given committee of 128 nodes has a dishonest $\frac{2}{3}$ majority is $4*10^{-15}$

\subsection{Bootstrapping}

In \sys, at the end of the previous epoch, based on a common random seed (e.g., Ethereum's RANDAO\cite{randao}), instead of dividing all $N$ validators into 32 committees as in Ethereum, we construct 32 distinct trees where we aggregate all $N$ signatures in a different tree in each slot. In each tree, all $N$ validators participate as \textit{leaf validators}, and additionally, a randomly selected subset of validators also acts as \textit{aggregators} within the protocol.

As outlined in the previous section, we assume that all validators know each other's IP addresses. However, as it is impractical to maintain a TCP connection to millions of validators, at the end of the current slot each validator only connects to their respective parent and child nodes of the next slot, while the block is being gossiped to the proposer.
Thus, validators can establish a TCP connection outside of the commit path of consensus, without impacting the system performance. %This also leaves validators sufficient time to configure their firewall rules for the next slot (more details in Section~\ref{sec:discussion}).
%Given that gossip takes up to 4 seconds~\cite{ethereum4sresearch}, and the leader still has to aggregate the signatures, build the new block, and gossip it, validators have sufficient time to set up the new connections for the next round.

\subsection{Protocol}

As previously outlined, there are 4 key roles in \sys: Leaf Validators, Leaf Aggregators, Internal Aggregators, and the Leader/Proposer.
In the following, we describe the algorithms for each of the roles.

\begin{algorithm*}[!t]
\centering
\caption{Leaf Validator}
\label{algo:leaf_validator}
\begin{algorithmic}[1]
\Function{deliver}{$B_i$} \Comment{Deliver Block $B_i$ from the Gossip Layer}
\State{$agg\_pub \gets aggregate(B_i.bit\_vec)$} \Comment{Aggregate up to $N$ public keys}
\If{$verify(B_i.agg\_sig, agg\_pub)$} \Comment{Verify Aggregate Signature}
\State{$send(vote(B_i))$} \Comment{Send vote to leaf aggregator committee}
\EndIf
\EndFunction
\end{algorithmic}
\end{algorithm*}

All $N$ validators in the system are \textit{leaf validators} in each round.
The leaf validator algorithm is described in Algorithm~\ref{algo:leaf_validator}. Leaf validators verify the aggregated signature of the previous round by retrieving the public key bit vector from the block $B_i$, aggregating the respective public keys, and then verifying the signature in the block with the aggregated public key. If the signature is valid, leaf validators sign the hash of this new block $B_i$ and send the resulting signature vote to their respective leaf aggregator committee.

\begin{algorithm*}[!t]
\centering
\caption{Leaf Aggregator}
\label{algo:leaf_aggregator}
\begin{algorithmic}[1]
\State{$agg\_sig \gets \bot$}
\State{$bit\_vec \gets 0$} \Comment{Empty Bitvector}
\Function{deliver}{$vote_i^j$} \Comment{Deliver vote $v_i^j$ for block $B_i$ from leaf $j$}
\If{$verify(vote_i^j)$}
\State{$bit\_vec[j] \gets 1$} \Comment{Note participation of $j$ in bitvector}
\If{$agg\_sig$ is $\bot$}
\State{$agg\_sig \gets vote_i^j.sig$}
\Else
\State{$agg\_sig \gets aggregate(agg\_sig, vote_i^j.sig)$} \Comment{Aggregate Signatures}
\EndIf
\If{$|bit\_vec| = m$}
\State{$send(agg\_sig, bit\_vec)$}
\EndIf
\EndIf
\EndFunction
\Function{$on\_timeout$}{} \Comment{After Timeout $\Delta$}
\State{$send(agg\_sig, bit\_vec)$}
\EndFunction
\end{algorithmic}
\end{algorithm*}

A smaller subset of validators has additional responsibilities, for example, being a \textit{leaf aggregator} as described in Algorithm~\ref{algo:leaf_aggregator}. Leaf aggregators receive votes from the leaf nodes and verify and aggregate them. After receiving all votes or a timeout, they send the aggregated signature and the respective public-key bit vector to the internal aggregators.

\begin{algorithm*}[!t]
\centering
\caption{Internal Aggregator \& Leader/Proposer}
\label{algo:internal_aggregator}
\begin{algorithmic}[1]
\State{$(agg\_pub\_vec,agg\_sig\_vec,bit\_vec\_vec) \gets (\emptyset,\emptyset,\emptyset)$}
\State{$(agg\_sig,bit\_vec) \gets (\bot,\emptyset)$}
\Function{deliver}{$agg\_sig^z, bit\_vec^z$} \Comment{Deliver agg\_sig and agg\_vec for subcommittee $z$}
\State{$agg\_pub^z \gets \bot$}
\For{$bit, index \in bit\_vec^z$}
\If{$bit$} \Comment{If public key is present (bit is set)}
\If{$agg\_pub$ is $\bot$}
\State{$agg\_pub^z \gets pki[z*m+index]$} \Comment{Get public key at offset and index}
\Else
\State{$agg\_pub^z \gets aggregate(agg\_pub^z, pki[z*m+index])$} \Comment{Aggregate public keys}
\EndIf
\EndIf
\EndFor
\If{$verify(agg\_sig^z, agg\_pub^z)$} \Comment{Verify Aggregated Signature}
\State{$bit\_vec\_vec[z] \gets bit\_vec\_vec[z] \cup bit\_vec^z$} \Comment{Add bit vector to set}
\State{$agg\_pub\_vec[z] \gets agg\_pub\_vec[z] \cup agg\_pub^z$} \Comment{Add aggregate public key to set}
\State{$agg\_sig\_vec[z] \gets agg\_sig\_vec[z] \cup agg\_sig^z$} \Comment{Add aggregate signature to set}
\If{$|bit\_vec\_vec[z]| = 16$} \Comment{If we collected all $16$ bitvectors of subcommittee $z$}
\State{$aggregate\_sub\_committee(z)$}
\EndIf
\If{$|bit\_vec| = \frac{m}{16}$} \Comment{Aggregated all child committee aggregates}
\State{$send(agg\_sig, bit\_vec)$}
\EndIf
\EndIf
\EndFunction

\Function{$aggregate\_sub\_committee$}{$z$}
\State{$(largest\_bit\_vec^z, largest\_index) \gets (\bot, \bot)$} \Comment{Track largest}
\For{$bit\_vec, index \in bit\_vec\_vec[z]$} \Comment{Find largest $bit\_vec$}
\If{$largest\_bit\_vec$ is $\bot$}
\State{$largest\_bit\_vec \gets (bit\_vec, index)$}
\ElsIf{$|bit\_vec| > |largest\_bit\_vec|$}
\State{$(largest\_bit\_vec,largest\_index) \gets (bit\_vec, index)$}
\EndIf
\EndFor
\State{$bit\_vec[z] \gets largest\_bit\_vec^z$}
\If{$agg\_sig$ is $\bot$} \Comment{Set or aggregate}
\State{$agg\_sig \gets agg\_sig\_vec[z][largest\_index]$}
\Else
\State{$agg\_sig \gets aggregate(agg\_sig, agg\_sig\_vec[z][largest\_index])$}
\EndIf
\EndFunction

\Function{$on\_timeout$}{} \Comment{After Timeout $\Delta i$(inverted depth)}

\For{$index \in 0..\frac{m}{16}$} \Comment{Iterate over all subcommittee indizes}
\If{$bit\_vec[index] = \bot$}
\State{$aggregate\_sub\_committee(index)$}
\EndIf
\EndFor
\State{$send(agg\_sig, bit\_vec)$}
\EndFunction
\end{algorithmic}
\end{algorithm*}

Both the \textit{Proposer} and \textit{internal aggregator} run the same protocol, with the exception that the proposer does not send the resulting aggregate to a higher-up committee, but instead forwards it internally to include it in its block (Algorithm~\ref{algo:internal_aggregator}).
Internal aggregators receive, on expectation, 16 signature aggregates per child committee $z$ for a total of $\frac{m}{16}$ child committees. They aggregate all public keys for each of the received signature aggregates, verify the signature, and store them temporarily.
On receiving 16 (or after a timeout), the internal aggregator selects the largest aggregate it received from each of the child committees $z$ and then aggregates them into a single aggregate which is then sent to the next internal aggregator, or internally forwarded to block creation.

% Part of bootstrapping is also finding the right fanout and stuff
% Need better title, but this is how we find a good fanout etc.
\smallskip\noindent{\bf Optimal Tree Identification.}
Before validators can set up the connections, they need to know the architecture and their role and position in the tree to identify their parent and child nodes.
As such, given $N$ validators, we want to deterministically identify the fanout $m$ that \emph{minimizes the time} it takes to aggregate $N$ signatures.
Importantly, while trees with smaller fanout $m$ distribute the computational load more evenly, they come at the cost of increased networking latency $\delta^{\text{tree}}$.

As a first step, we want to calculate the time $\delta$ it takes to aggregate $N$ signatures in a generic tree $T'$ considering $N$ validators and a generic fanout $m$. $\delta$ is a combination of the computational load $\delta^{\text{comp}}$ (e.g., verification and aggregation), and the network delay $\delta^{\text{tree}} = \Delta\cdot d$ given the tree depth $d$ and network latency $\Delta$.
To simplify the analysis, we assume that validators verify all messages they receive and aggregators only include the largest aggregate from each child committee in their final aggregate.
Furthermore, we divide our tree into two subtrees. One tree of depth $d=1$ and fanout $m$ and one tree of fanout $m'=\frac{m}{16}$(given 16 committee representatives) of depth $d=\ceil{\log m'(\frac{N}{16})}$. Resulting in a total depth of $d=\ceil{\log m'(\frac{N}{16})}+1$.

Given a perfectly balanced tree, to calculate $\delta^{\text{comp}}$, we consider the public key aggregation cost $\mathsf{pka}$, the signature aggregation cost $\mathsf{sga}$, and the signature verification cost $\mathsf{sgv}$.
Considering the four validator roles in the system, we need to calculate the cost of the individual types at each layer of the tree.

First, leaf validators aggregate $N$ public keys to verify the signature included by the proposer and execute the block to verify validity which takes $\delta^{\text{execute}}$, resulting in $N \cdot \mathsf{pka}+\mathsf{sgv}+\delta^{\text{execute}}$. Next, leaf aggregators, have to verify and aggregate $m$ signatures, resulting in $m \cdot \mathsf{sgv}+m \cdot \mathsf{sga}$.

After this, internal aggregators receive $m$ messages for which they have to aggregate the respective public keys to verify the signature aggregates, and select a subset of size $m'$ to aggregate and forward.
The number of public keys that have to be aggregated depends on the depth of the tree. Given the inverted depth $i$ (i.e., $i=d$ at the proposer and $i=0$ at the leaves), $m\cdot m'^i \cdot 16$ public keys are aggregated per child committee in a perfectly balanced tree.
This results in the total cost of $\sum_{i=1}^{d+1}{m \cdot \mathsf{sgv}+m' \cdot \mathsf{sga}+\min(N, m \cdot m'^i) \cdot 16 \cdot \mathsf{pka}}$.

Finally, for any perfectly balanced tree $T'$ with fanout $m$ and $N$ validators, this results in a total of: $\delta^{\text{agg}} = \sum_{i=1}^{d+1}{m \cdot \mathsf{sgv}+m' \cdot \mathsf{sga}+\min(N, m \cdot m'^i) \cdot 16 \cdot \mathsf{pka}} + m \cdot \mathsf{sgv}+m \cdot \mathsf{sga} + N \cdot \mathsf{pka}+\mathsf{sgv}$. Given $C$ cores, most of the aggregation and verification load can be parallelized, resulting in $\delta = \frac{\delta^{\text{agg}}}{C} +\delta^{\text{execute}} + \Delta d$.
Based on this, given $N$, we can find the best tree $T'$ by calculating $\delta$ for all possible $m$. We discuss efficient approaches to identify the best trees in Appendix~\ref{sec:discussion}.

\smallskip\noindent{\bf Reward Allocation.}\label{subsec:rewards}
We allocate attestation rewards over length-$k$ windows (fixed-disjoint blocks of $k$ slots, or sliding). A validator earns rewards only if it is included in at least one of the $k$ slots within a window, otherwise, it receives no reward for that window. As \sys aggregates all $N$ signatures every slot, validators have 32 opportunities per epoch, compared to just one in Ethereum.
Given this rule, the probability that an honest validator loses a window's reward is $(1-p)^k$, where $p$ is the per-slot inclusion probability. We choose $k$ to keep this loss negligible, while keeping $k$ small enough to discourage free-riding (see~\Cref{sec:discussion} for more details and a discussion on free-riding).

\subsection{Optimizing Public Key Aggregation}
\label{sec:optimizations}

One of the main cost factors of the protocol is the aggregation of all $N$ public keys.  
However, given that Ethereum and other proof-of-stake chains maintain an average participation rate above 99\%~\cite{cobra}, we can significantly optimize this.

We describe a straightforward optimization that leverages a unique property of BLS signatures: the ability to subtract signatures and public keys from an aggregate~\cite{ethbls}.  
Given a 99\% participation rate, instead of aggregating 99\% of public keys, we can cache the complete aggregate and subtract the 1\% missing public keys.  
Because subtraction corresponds to adding the inverse of the public key, subtracting a public key from an aggregate has the same computational cost as addition. Thus, this optimization achieves a $99\%$ reduction in the cost of public key aggregation.

Even in the worst case, with up to $f \approx \frac{1}{3}$ faulty validators, we only need to subtract the $\frac{1}{3}$ missing public keys instead of aggregating the remaining $\frac{2}{3}$, effectively halving the computational cost.  
In general, for a given participation rate $r$, this optimization reduces the aggregation cost by a factor of $\frac{1 - r}{r}$. 
We show in~\Cref{sec:varying} that our optimization significantly reduces system load even in the case of 2.5 million nodes in the worst case setting of $\frac{2}{3}$ participation rate.

However, this subtraction-based approach is more difficult to parallelize.  
To address this, we can divide the cached public key aggregate into $C$ chunks, allowing $C$ parallel tasks to independently remove missing public keys from each chunk.  
This requires performing $C$ additional aggregations, which is negligible when $N$ is large, as we expect $C$ to remain small (e.g. $C \leq 16$).

\section{Censorship-Resilience in \sys }\label{sec:analysis}

We analyze \sys in the Byzantine model with a slowly adaptive adversary (Section~\ref{sec:systemmodel}).
Under our reward allocation policy (\Cref{subsec:rewards}), it suffices to bound the probability that an honest validator's vote cannot be censored for $k$ consecutive slots.
Guaranteeing that this probability is small ensures honest participation is rewarded and stake does not drift toward the adversary.
To this end, we first lower-bound the per-slot inclusion probability $p=p(d,m)$ as a function of tree depth $d$ and fanout $m$, and then lift this bound to $k$ consecutive slots.

\smallskip\noindent{\bf Minority correct committee leadership.}
We consider vote censorship attack opportunities involving the internal aggregators of the tree. We first observe that deeper aggregation trees have more vote censorship opportunities involving the internal aggregators.
As internal aggregators only forward the largest aggregate they receive, malicious nodes can collude and guarantee that the largest aggregate always originates from a malicious aggregator.

In more detail, malicious leaf nodes send their votes exclusively to malicious aggregators.
As long as there is at least one malicious aggregator and $x\geq 2$ malicious leaf nodes, the malicious aggregator can construct a larger aggregate than any correct aggregator by always including votes from the $x$ malicious leaves that no correct aggregator received.  
This strategy allows the malicious aggregator to produce a larger aggregate than any honest aggregator, while censoring up to $x-1$ correct nodes. As such, this allows malicious aggregators to censor up to $\frac{1}{3}$ of the committee members.
%$\approx \frac{1}{3}$ of the committee members.

Suppose we have $f$ total adversarial nodes and $L$ leaf groups. 
We now compute the probability that there is one leaf group with $\geq 3$ adversarial nodes, with one of the adversarial nodes becoming one of the 16 aggregators.

\begin{lemma}\label{lem:successprob}
    The probability of at least one leaf group with $\geq 3$ adversarial nodes and $\geq 1$ adversarial aggregator is at least $ \sum_{s=3}^{128-16} \left( \frac{\binom{f}{s}\binom{N-f}{128-16-s}}{\binom{N}{128-16}} \right) \cdot \sum_{s=1}^{16}\left(\frac{\binom{f}{s}\binom{N-f}{16-s}}{\binom{N}{16}} \right) $
\end{lemma}
\begin{proof}
    Follows from the Hypergeometric distribution. 
\end{proof}

With $N = 1000000$, $f = \frac{N}{3}$, the above probability is $0.998$ for a single round and with $f=0.05N$, the above probability still remains fairly large at $0.548$.

Although we could avoid censorship in this situation by including all aggregates in the final aggregate, this would significantly increase the computation, storage, and networking costs.
However, we can significantly reduce the probability of success of this attack by additionally choosing a random aggregate alongside the largest.
For a tree of depth $d$, the following lemma computes the probability of a single leaf node being successfully included in a single slot with $f \leq \frac{1}{3}$ faulty nodes. The proof of~\Cref{lem:internal} is in~\Cref{app:proof_internal}

\begin{restatable}[]{lemma}{internal}\label{lem:internal}
    For a tree of depth $d$, the probability of a single leaf node being successfully included in a single slot is at least $\frac{2}{3} \cdot \frac{2}{3}^{d-2}\cdot \frac{16-\frac{16}{3}}{15}$
\end{restatable}

Let us denote the probability as specified by~\Cref{lem:internal} as $p$.
We recall that there is only a penalty if a node consistently fails to participate for $k$ consecutive slots, and this probability is $(1-p)^k$ for any single node.
Considering there are many concurrent committees that the malicious validators can attempt to censor in parallel, given $L$ committees on the lowest level of the tree, the probability for no committee to be censored is $(1-(1-p)^k)^L$ over $k$ rounds.

Concretely, for $k=64$ (2 epochs), $m=256$ and $L=4096$ with a depth of $4$, this results in a $0.998$ probability that no node got censored consistently over two epochs.
We stress that this is a much higher guarantee compared to the status quo on Ethereum, where, even if it followed the 2 epoch rule it would only ensure a $88.9\%$ guarantee for each node. 
To achieve comparable vote censorship resilience guarantees, in the current system, Ethereum would require a 6 epoch timeout before a node can lose out on rewards due to the lack of participation.

\begin{figure}[ht]
    \centering
    \includegraphics[width=\linewidth]{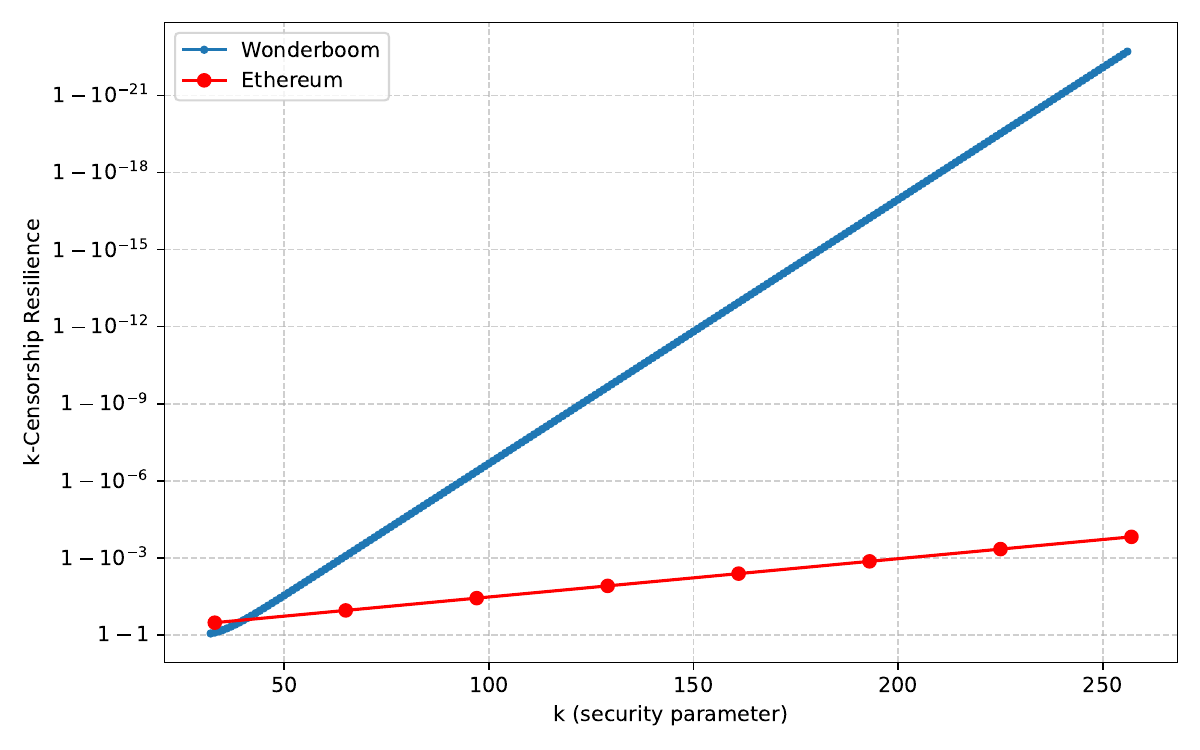}
    \caption{Comparison of censorship resilience of Wonderboom and Ethereum for different values of $k$.}
    \label{fig:prob}
\end{figure}

Figure~\ref{fig:prob} illustrates the difference in vote-censorship resilience between the two approaches in the presence of $1/3$ faulty validators. The x-axis represents the security parameter $k$, while the y-axis shows censorship resilience (the inverted probability of a successful censorship attack). \ray{In Ethereum, this probability is dominated by two factors: (i) $1/3$ probability to be in a slot with a faulty proposer, and (ii) one opportunity to vote every 32 slots. Thus, for $k$ consecutive epochs this results in $(1/3)^k$.}
Compared to Ethereum, \sys not only achieves significantly higher censorship resilience, but this resilience also grows much more rapidly as $k$ increases.

%\smallskip{\em Faulty Block Proposer.}
% They can just censor any subtree. Easiest attack target. But we don't know their IP!
%The validator only has to corrupt 1 out of the aggregator pools on the way to the leader. So they have multiple opportunities. So it is $(1-0.687)^{d-2}*\frac{1}{3}$ given there is a $\frac{1}{3}$ probability that the leader is faulty and excludes a given sub committee aggregate.
%This results in $sth$, which for several rounds, is already negligible.

\smallskip\noindent{\bf Impossibility against fully adaptive adversaries}
To motive our main analysis and result of the vote-censorship resilience of \sys under slowly adaptive adversaries, we first observe that it is impossible to achieve vote censorship resilience against a fully adaptive adversary, i.e., an adversary that can corrupt validators \emph{after} observing the validator placements in the tree structure at the beginning of each slot. 
We stress that the result in~\Cref{thm:strongadaptive} is not a weakness of \sys but a general weakness of all single proposer based protocols against fully adaptive adversaries, whereby the adversary can simply corrupt the proposer if it knows the identity of the proposer and is allowed to corrupt it at the start of each round.

\begin{theorem}\label{thm:strongadaptive}
\sys is insecure under fully adaptive adversaries.
\end{theorem}

\begin{proof}
    We construct a strategy that enables the adversary to conduct an almost surely successful vote-censorship attack.
    Since the adversary can corrupt up to $f$ validators for the whole epoch, with the choice to adaptively select the validators at the start of each slot after seeing the validator placements, the adversary just needs to corrupt the root of \sys. 
    This means that as long as $f\geq32$, the adversary can almost surely succeed in censoring up to $\frac{1}{3}$ every slot. 
    Under this strategy, \sys is clearly insecure. 
\end{proof}

\smallskip\noindent{\bf Towards general vote-censorship resilience.} 
Although~\Cref{thm:strongadaptive} states an impossibility result against fully adaptive adversaries, we can achieve vote-censorship resilience if we consider a weaker, slowly adaptive adversary. 
We first show in~\Cref{thm:adaptive} in~\Cref{app:proof_adaptive} that the optimal corruption strategy for a slowly adaptive adversary that can corrupt up to $f$ validators per epoch is to sample $f$ validators uniformly at random to corrupt.
From~\Cref{thm:adaptive}, we show that under per-slot re-randomization of committee placements and epoch-boundary corruption, the worst-case strategy is to corrupt $f$ validators uniformly at random at the start of the epoch. Combining this with~\Cref{lem:internal} yields the main result.
% Our main theorem,~\Cref{thm:main} states that \sys is vote-censorship resilient against slowly adaptive adversaries, i.e., adversaries that can make decisions on who to corrupt at the beginning of the epoch \emph{before} observing the validator placements in the tree structure for each slot in the epoch.
% The proof of~\Cref{thm:main} is in~\Cref{app:proofmain}, and the proof hinges on an intermediate result that states that a slowly adaptive adversary has no better corruption strategy than to randomly corrupt up to $f$ validators at the beginning of the epoch (see~\Cref{thm:adaptive} in~\Cref{app:adaptive}). We expound on this in more detail in~\Cref{app:adaptive}. 
% The intermediate result of~\Cref{thm:adaptive} together with the inclusion probability as computed in~\Cref{lem:internal} consequently gives us the proof of our main theorem. 

\begin{restatable}[\sys vote-censorship resilience]{theorem}{main}\label{thm:main}
    \sys is vote-censorship resilient against slowly adaptive adversaries.
\end{restatable}

\begin{proof}
    From~\Cref{thm:adaptive}, the optimal strategy for adaptive adversaries is to randomly corrupt a set of $f$ validators at the start of the epoch, before knowing their placements in the trees in each slot as defined by \sys.
    From~\Cref{lem:internal}, we get that with $L$ leaf groups, the probability that no leaf group is censored over $k$ consecutive rounds is $(1-(1-p)^k)^L$. 
    %Since $L$ is fixed, choose $k'=e^L$ and thus $(1-(1-p)^k)^L \rightarrow 1$ for $k > k'$. 
    Thus, \sys satisfies vote-censorship resilience as defined in~\Cref{def:er}. 
\end{proof}

\section{Implementation \& Evaluation}
\label{sec:evaluation}

We implemented a simulation tool for \sys to realistically evaluate its worst-case performance on the scale of millions of nodes. The simulator is written in Rust and implements the essential functionality of \sys, including communication, signature verification, aggregation, and protocol logic\footnote{Available at: https://anonymous.4open.science/r/Wonderboom-main and
https://anonymous.4open.science/r/Wonderboom-ethereum/README.md}.

To evaluate \sys at scale, the tool executes the actual protocol logic for one node per role in sequence while simulating the responses of all other nodes in the system. This approach allows modeling networks with millions of participants on a single machine. We use rayon thread pools to control parallelism and the widely adopted \textit{blst}~\cite{blst} library for BLS signature support. However, we implemented signature subtraction since it was not natively supported by \textit{blst} at the time of writing.

In our evaluation, to simplify the worst-case analysis, we assume that validators must verify all fanout $m$ messages before they can filter out invalid votes, and that all $N$ votes are transmitted even though only a $\tfrac{2}{3}$ fraction is valid. This gives us an upper bound on the worst-case performance.

We ran the experiments on an AMD Ryzen 9 6900HX system with 32 GB of RAM. Throughout all experiments, memory usage remained below 8 GB.
\ray{To emulate network latency, we used the Linux \textit{netem} tool~\cite{netem} with a fixed latency of $100\mathrm{ms}$ between participants (i.e. 200ms RTT).}

The evaluation has three primary goals. First, we aim to determine the configurations under which \sys can aggregate all $N$ signatures within a single slot (4 seconds) even in the worst case. Second, we investigate how relaxing the worst case, e.g., by assuming higher participation rates, affects the practical load on nodes. Third, we want to show how the performance of \sys compares to the state of the art algorithm on Ethereum if all the signatures would be aggregated in a single slot.

    \begin{figure}[t]
        \centering
        \includegraphics[width=\linewidth]{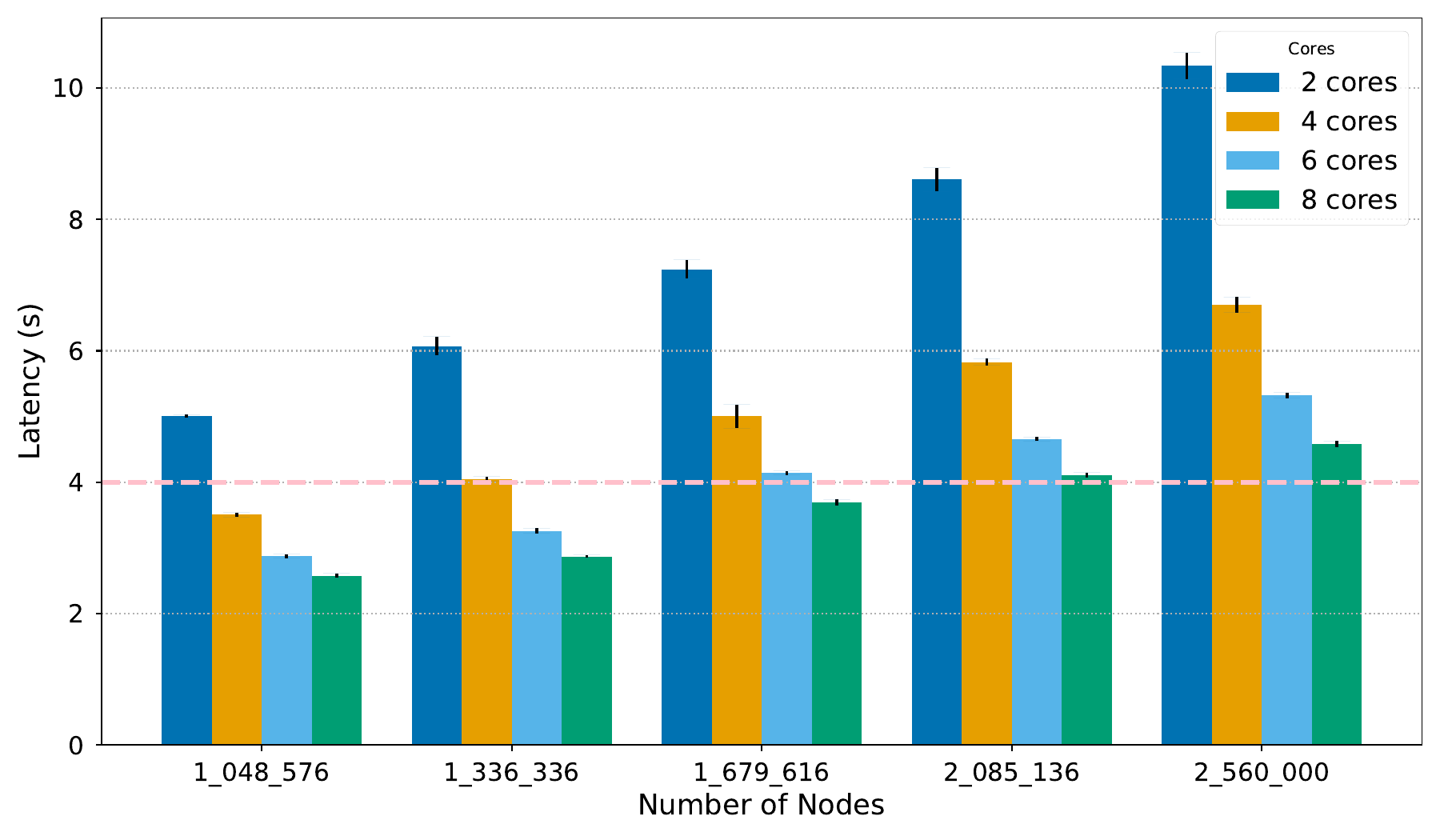}
        \caption{Worst-Case Configurations}
        \label{fig:architecture_nodes}
    \end{figure}%
    \begin{figure}[t]
        \centering
        \includegraphics[width=\linewidth]{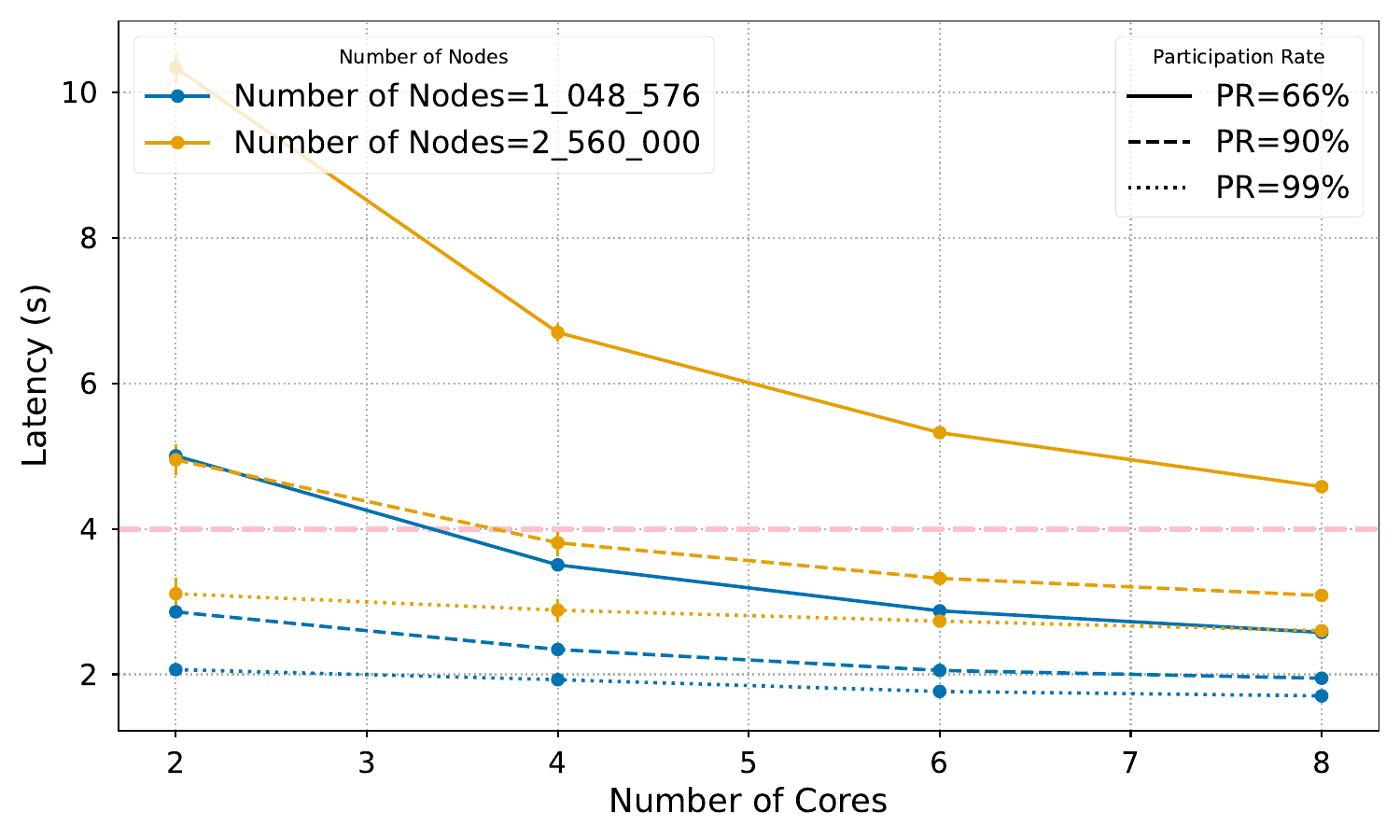}
        \caption{Runtime of \sys under varying participation rates.}
        \label{fig:architecture_cores}
    \end{figure}
    \begin{figure}[t]
        \centering
        \includegraphics[width=\linewidth]{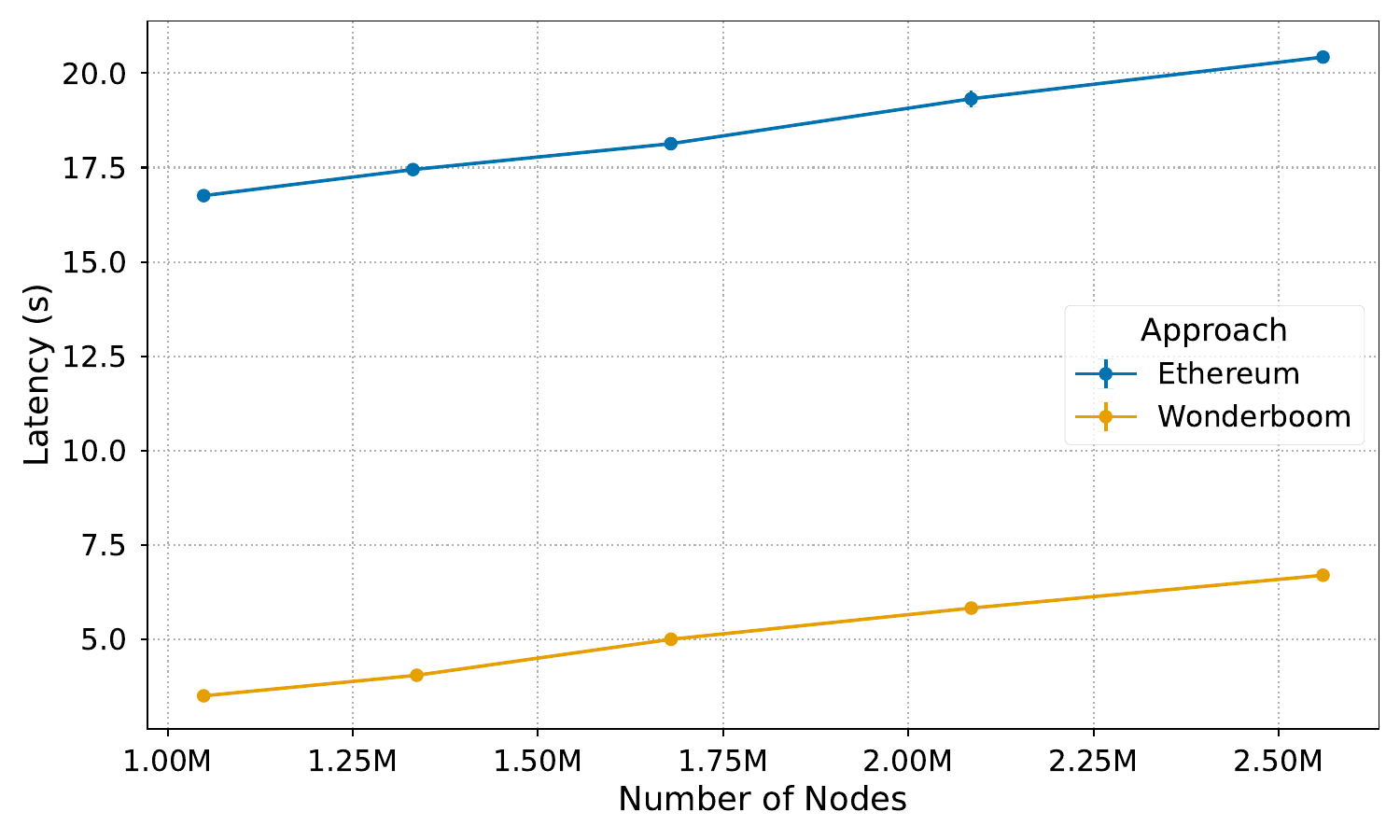}
        \caption{\sys vs Ethereum}
        \label{fig:ethwonder}
    \end{figure}%

\smallskip\noindent{\bf Worst Case.}
Figure~\ref{fig:architecture_nodes} shows how different configurations of \sys perform under worst-case conditions. The y-axis represents latency in seconds with a cut-off at 4 seconds, while the x-axis shows the number of nodes alongside the available cores. The number of nodes ranges from 1 million to 2.5 million (fanout 256 to 320). For the current Ethereum configuration of about 1 million validators, all $N$ signatures can be aggregated in under 4 seconds with only 4 cores, which matches the minimum Ethereum node requirement. For larger configurations (1.5 to 2 million nodes), aggregation remains possible with 8 cores, only at 2.5 million nodes, aggregation takes a little bit longer than 4 seconds.. These results indicate that \sys can scale Ethereum to much larger validator sets while still guaranteeing two-slot consensus, even under worst-case conditions.

\smallskip\noindent{\bf Varying Participation Rate.}\label{sec:varying}
Next, we evaluate how the runtime of \sys behaves under higher participation rates. We still assume worst-case verification and networking costs at each rate, omitting all optimizations except for optimizing the public key aggregation based on the higher participation rate as described in Section~\ref{sec:optimizations}.

Figure~\ref{fig:architecture_cores} illustrates how runtime and node requirements change as participation increases. The y-axis shows latency in seconds, while the x-axis shows the number of cores. We simulated two validator set sizes: 1 million (blue) and 2.5 million (orange). The solid line indicates the $\frac{2}{3}$ minimum participation threshold, while the dashed and dotted lines represent 90\% and 99\% participation respectively (reflecting current Ethereum conditions~\cite{cobra}).

Our results show that optimizing public key aggregation substantially reduces system load, confirming that public key aggregation is the dominant performance factor. At 90\% participation, \sys can handle 2.5 million nodes with 4 cores and 1 million signatures with 2 cores. At 99\% participation, both configurations complete within 4 seconds using only 2 cores.
This indicates that even legacy Ethereum nodes with weaker node hardware will be able to support two slot consensus on the Ethereum blockchain with the help of \sys.

\smallskip\noindent{\bf Wonderboom vs Ethereum.}\label{sec:evalcomparison}
Figure~\ref{fig:ethwonder} compares \sys to the current Ethereum design, with the number of nodes on the x-axis and latency in seconds on the y-axis. Both configurations were run on 4 cores and with the worst-case participation rate, i.e., liveness-critical conditions.

First, we observe that \sys substantially outperforms Ethereum, requiring less than a third of the time in all configurations validating our design and demonstrating the necessity of \sys to achieve shorter finality, particularly at larger scales.

Note that this simulation of the Ethereum aggregation protocol also requires point to point channels between validators. However, in practice, in Ethereum all communication goes through the Gossip layer where each message takes around 4 seconds to reach all validators~\cite{ethereum4sresearch}. One of the main reasons for this large delay is the large number of individual signatures (over 32.000) that are disseminated all at once through the gossip layer.
Therefore, in practice, Ethereum is not only limited by the performance of the aggregation protocol but also by the performance of the underlying gossip protocol. As a result, while the aggregation protocol is already significantly faster compared to the current protocol in Ethereum, by using point to point channels instead of the gossip network for signature dissemination, we can aggregate all $N$ signatures over 32x faster than the current Ethereum protocol.

%Finally, it is important to note that in this example, communication between nodes in the shallow tree was peer-to-peer rather than through gossip. In Ethereum, tree communication relies on gossip, which substantially limits scalability. At the time of writing, the gossip layer already introduces significant delays, which would only worsen if more nodes attempted to broadcast their signatures each round.}

\smallskip\noindent{\bf Summary.} Our evaluation highlights the advantages of \sys over Ethereum's current approach. We demonstrate that \sys can aggregate signatures for large validator sets within a single Ethereum slot, even under worst-case conditions with up to $\frac{1}{3}$ Byzantine nodes, and that during periods of higher participation this can be achieved even on legacy hardware. Furthermore, we show that \sys exhibits fundamentally better scalability properties than the existing protocol, making it a strong candidate as Ethereum plans to lower the minimum deposit requirement, which could substantially expand the validator set. Consequently, \sys not only enables significantly shorter finality latency but also enables operation at much larger scales.

% blah
\section{Discussion}
\label{sec:discussion}
In this section, we discuss practical considerations, design trade-offs, and optimizations of \sys, and their implications for performance, security, and deployability.

\smallskip\noindent{\bf Random Aggregate Overhead.}
In \sys, to achieve vote-censorship resilience, aggregators submit two aggregates, the largest aggregate and a random aggregate.  
This doubles the bandwidth and metadata storage requirements but is primarily only necessary in the worst-case scenario when the system is under attack.  
In practice, given the high participation rate of Ethereum, the random aggregate will usually also match the largest aggregate.  
In such cases, the aggregator can submit only a single aggregate, bypassing this overhead.

\smallskip\noindent{\bf Roaring Bitmap.}
Unlike Ethereum, \sys aggregates signatures from all $N$ validators every round, requiring each block to include a bitset of all validators and increasing metadata overhead by a factor of 32. However, due to Ethereum's high participation rate, these bitsets can be efficiently compressed, for example using Roaring Bitmaps~\cite{roaring}.

\smallskip\noindent{\bf Separate Validator/Proposer IP.}
As discussed in Section~\ref{sec:wonderboom}, we propose physically separating the validator and proposer processes by assigning them distinct IP addresses. The key motivation for this architectural decision is that it enables validator IP addresses to be publicly known while still protecting the proposer from censorship attacks.

In the current Ethereum design, the proposer is protected as all communication passes through the Ethereum gossip layer, which aims to obfuscate the proposer's identity. However, prior work has shown that because validators interact with the gossip layer at least once per epoch, their IP addresses can be deanonymized, which in turn also exposes the proposer~\cite{heimbach2024deanonymizing}.

As such, our design enables direct TCP communication between validators, bypassing gossip-layer latency, while also reducing proposer exposure to known deanonymization attacks. Because the proposer interacts with the gossip network only rarely in this architecture, these attacks become significantly harder to carry out in practice. Additionally, reducing gossip traffic can lower the overall network load, which could help mitigate the observed tail latency of up to 4 seconds.

%The motivation is to protect the proposer from censorship attacks, which, e.g., could be incentivized by MEV opportunities. In the current Ethereum design, the proposer's identity can be uncovered over time, as attackers can gradually deanonymize validator IP addresses~\cite{heimbach2024deanonymizing}. By assigning a separate IP to the proposer, the deanonymization process becomes significantly harder, as the proposer only interacts infrequently with the gossip layer.

A similar separation of roles exists in blockchains like Aptos~\cite{aptos}, where validators are encouraged to operate a full node alongside the validator to handle client requests and shield the validator from direct exposure. However, this approach incurs a substantial increase in hosting and operational costs.

In contrast, our design aims to minimize both cost and complexity. Rather than requiring a separate full node, the validator may be implemented as a lightweight gateway with a dedicated IP address that receives messages from other validators and forwards them to the proposer. This gateway-based architecture is a common design pattern and was, for example, widely used by mining pools in PoW Ethereum~\cite{gateways}.

%IPv6 addresses are inexpensive, and validators can reserve additional addresses to separate proposer and validator processes at negligible financial overhead. The validator process in our model would not hold keys or perform heavy computation. Instead, it would act as a lightweight forwarding service, configuring firewall rules to accept connections only from the validators relevant in the current and next round. It would then proxy valid requests to the proposer process and relay responses back to the network.
%This design reduces the proposer's exposure to attacks compared to Ethereum's current model, while enabling \sys to leverage direct point-to-point communication between validators. This also reduces the load on Ethereum's gossip network, further improving the protocol latency.

\smallskip\noindent{\bf Free Riding.} In \sys, each validator must participate at least once every $k$ slots to remain eligible for rewards. However, this also implies that validators may remain idle for up to $k - 1$ slots without immediate penalty. 
As we show in Section~\ref{sec:analysis}, for $k = 2 \times 64$ (two epochs), we achieve similar censorship resilience to that of six epochs in Ethereum. Since Ethereum currently requires only one participation per epoch, this highlights a fundamental trade-off between censorship resilience and resistance to free-riding behavior.

\ray{
However, as shown in~\cite{badslashing}, this trade-off is inherent to Proof-of-Stake blockchains, and slashing mechanisms intended to deter free riding can undermine fundamental protocol properties.
}

%Consistent participation could be encouraged by adopting mechanisms such as those proposed in~\cite{cobra}, where participation rewards are distributed only if finality is achieved within two slots, and the total reward pool is scaled according to the overall participation rate. Thus, if only $\frac{2}{3}$ of the validators participate in a given slot, only $\frac{2}{3}$ of the rewards are issued. While this approach provides only a marginal additional incentives for active validators, we argue that the potential benefit from abstaining is limited, as the main operational costs of running a validator cannot be significantly reduced by sending attestations less frequently. A comprehensive analysis of these trade-offs is left for future work.

\smallskip\noindent{\bf Aggregate-Verify.}
After public key aggregation, the next largest cost factors in \sys is that each aggregator must verify $m$ BLS signatures individually.  
However, aggregators in \sys could wait until receiving a sufficient number of signatures, aggregate them, and then verify only the aggregate.  
This would substantially reduce the average-case verification cost.  

However, in the worst case, when up to $\approx \frac{1}{3}$ of nodes are faulty, aggregate verification may fail due to a single invalid signature, forcing a fallback to sequential verification.  
This creates a trade-off: improved average-case complexity at the expense of worse worst-case complexity.  
We omit this optimization when presenting our design, as our goal is to ensure that the protocol can aggregate millions of signatures within a single slot even in the worst case.

\smallskip\noindent{\bf Identifying Tree Configurations Efficiently.}
% We identify two main approaches to find the best possible trees for different system configurations. First, we can approach it with a greedy algorithm, as, even though there are $O(N^2)$ combinations of $N$ and $m$, in practice, this only has to be calculated once every time the minimum system requirements change, and the best configurations for different $N$ can then be hardcoded in the validator software. Furthermore, we only have to calculate practical values for $m$ and $N$, e.g., $m$ can be limited to $\log N < m < \sqrt N$ and $N$ only incremented in $\sqrt N$ steps.
% Second, we can also use linear programming to determine the most efficient tree. 
% This computation only needs to be performed once, and thus, we defer the development of more optimized algorithms to future work.
Determining the optimal tree structure is central to ensuring both scalability and robustness. We outline two complementary approaches to selecting optimal tree parameters across different system configurations. A straightforward option is a greedy search: although the number of candidate configurations is quadratic in $N$ and $m$, in practice this computation is required only once whenever minimum system requirements are updated. The resulting optimal configurations for different $N$ can then be hardcoded into validator software. Moreover, the search space can be restricted to practically relevant ranges, such as $\log N < m < \sqrt{N}$, with $N$ varied in increments of $\sqrt{N}$.  

A second approach is to formulate the problem as a linear program, allowing the most efficient tree structure to be computed directly. Since this optimization must be carried out only once, we defer the development of more sophisticated or adaptive algorithms to future work.

\smallskip\noindent{\bf Geographic Self-Selection.}
\label{app:geo}  
In both \sys and Ethereum, validators are randomly assigned to roles and positions in the tree. Consequently, the time required for leaf aggregators to collect signatures depends on the ``{global}'' latency.  
One natural idea to reduce this latency is to cluster validators so that leaf nodes and their aggregators are geographically close, thereby minimizing network delay. While this would allow even poorly connected validators to participate more reliably within tight slot deadlines, it introduces two significant drawbacks. First, it requires knowledge of network latencies between individual validators, which complicates protocol design. Second, and more critically, it makes the system more vulnerable to targeted attacks: an adversary could concentrate efforts on a single geographic cluster instead of the entire network.  

We therefore propose \emph{geographic self-selection} as a lightweight alternative. In its simplest form, each validator chooses an identifier, and the protocol batches validators that share the same identifier. Through social consensus, identifiers can be associated with geographic regions (e.g., continent, country, city), but the protocol itself need not interpret their meaning. More flexibility can be achieved through ranked or hierarchical identifiers (e.g., continent at the top level, then country, then region), enabling batching at the lowest level with sufficient participants.  

This design, however, comes with censorship risks: an attacker could register multiple validators under the same identifier as a target, thereby controlling a large fraction, or even all, of the aggregators in that cluster. To mitigate this, validators could be reassigned to random positions if they fail to participate for several consecutive rounds, and allowed to change identifiers frequently if their current choice proves disadvantageous. These safeguards reduce the risk of persistent censorship but increase the security parameter $k$ required to ensure strong vote-censorship resilience.

\section{Related Work}
\label{sec:relatedwork}

There have been numerous works investigating eclipse attacks on the Ethereum network, e.g.,~\cite{eclipse1,eclipse2}. To the best of our knowledge, these works have primarily focused on the properties of the underlying gossip network rather than on Ethereum's signature aggregation protocol.

In contrast, tree-based vote aggregation has been a long-standing technique in consensus protocols, where it has been used to reduce computational and bandwidth overhead by distributing the aggregation process among participants. An example is Byzcoin~\cite{byzcoin} which distributes computational and bandwidth costs across nodes to accelerate large-scale consensus. 
%Early designs were vulnerable when faced with Byzantine validators. 
Later, systems such as Omniledger~\cite{omniledger} introduced mechanisms to tolerate Byzantine faults within the tree structure. Kauri~\cite{kauri} went further, proposing a reconfiguration algorithm that finds a robust tree in linear time, ensuring consensus despite faulty participants. More recent work~\cite{CTT} strengthens these resilience guarantees by introducing redundancy into the tree: instead of a single internal node forwarding leaf votes, two internal nodes forward them.
However, these approaches are optimized for networks with at most thousands of validators. Applied to ecosystems such as Ethereum, with over a million validators, they would open the door for censorship, eclipse and liveness attacks.

The scale of Ethereum has motivated designs such as SFF~\cite{sff} and FFG~\cite{ffg}, which propose requiring clients to verify and aggregate millions of consensus messages to detect finality.
This not only puts a large computational load on the clients, but also floods the Ethereum gossip layer with millions of messages.

In a 2022 blog post, Vitalik proposed several approaches to move Ethereum towards single-slot finality, one of which was adopted in the Pectra upgrade~\cite{pectra}, aiming to reduce the effective validator set by consolidating validators through weighted voting. In theory, this could shrink the active validator set to allow single slot vote aggregation. However, despite the introduction of weighted voting, the number of validators has only slightly decreased~\cite{beaconchain}. Moreover, if Ethereum lowers the entry barrier for becoming a validator as proposed, the validator set would grow further.
Another proposal was to use super-committees to finalize blocks. However, this reduces the economic safety of the protocol, as a smaller number of validators, with a smaller combined stake, would be responsible for finality~\cite{cobra}.

Similar vulnerabilities arise in quorum-based protocols that achieve scalability by sub-sampling the validator set and running consensus on small committees, such as Algorand~\cite{algorand} and Avalanche~\cite{avalanche}.
% Examples include Algorand~\cite{algorand}, which uses cryptographic sortition, and Avalanche~\cite{avalanche}, a gossip-based protocol with sub-sampled voting. 
% However, these approaches trade scalability for reduced economic safety. In Ethereum, successful attacks require control over a large fraction of the total stake, which can subsequently be slashed and potentially used to compensate for economic damage~\cite{cobra}. In contrast, committee-based protocols can only hold the misbehaving committee members accountable, whose combined stake is typically negligible.
These approaches reduce per-round communication and latency by limiting the number of participating validators, but in doing so, 
% decouple safety from the total stake securing the system. As a result, 
accountability is restricted to the selected committee, whose combined stake is typically a small fraction of the total, reducing the economic cost of attacks.
On the contrary, successful safety violations in Ethereum require control of a large fraction of the total stake, which can subsequently be slashed and potentially used to compensate for economic damage~\cite{cobra}. 

Taken together, existing approaches either sacrifice economic safety or fail to scale aggregation to Ethereum’s validator set.
% In this work, we introduce \sys, moving Ethereum towards 2-slot finality by aggregating the signatures of all correct validators within a single slot. Notably, \sys achieves this without reducing the validator set and without compromising security. 
We introduce \sys, which enables 2-slot finality by aggregating the signatures of all correct validators within a single slot, without reducing the validator set or compromising security.
In fact, \sys provides stronger vote-censorship resilience than the current state-of-the-art in Ethereum.
\section{Conclusion}\label{sec:conclusion}
% This paper introduced \sys, a tree-based signature aggregation protocol for large-scale consensus. 
% We evaluated its worst-case performance and demonstrated that \sys can aggregate signatures in a single slot for the current Ethereum validator set on legacy validator hardware (e.g. 4 cores), and even scale to 2 million validators with the current minimum hardware requirements. 
% Unlike existing protocols, \sys achieves this without introducing security trade-offs and provides higher censorship resilience guarantees than Ethereum.

This paper introduced \sys, a tree-based signature aggregation protocol for large-scale consensus. 
We evaluated its worst-case performance and demonstrated that \sys can aggregate signatures within a single slot for the current Ethereum validator set on legacy validator hardware (e.g., 4 cores), while scaling to validator sets of up to 2 million nodes under the same hardware assumptions. 
Unlike existing approaches, \sys achieves single-slot aggregation without reducing the validator set or introducing security trade-offs, and provides stronger vote-censorship resilience guarantees than Ethereum’s current design. 
Together, these results show that fast finality and strong inclusion guarantees are compatible with million-validator consensus.

\newpage

\appendix
\section{Acknowledgments}

The work was partially supported by the Austrian Science Fund (FWF) through the SFB SpyCode project F8509-N and F8512-N, and by the WWTF through the projects 10.47379/ICT22045 and 10.47379/ICT25056.
\section{Open Science}

The artifacts are available at: \url{https://anonymous.4open.science/r/Wonderboom-main} for the Wonderboom aggregation protocol and the instructions to reproduce our results. Furthermore, at \url{https://anonymous.4open.science/r/Wonderboom-ethereum} is our implementation of the Ethereum aggregation protocol we used to compare the two approaches.

\section{Ethical Considerations}

This paper discusses two vulnerabilities in Ethereum that have already been identified in prior work. First, the potential to deanonymize validators at the gossip layer was demonstrated in~\cite{heimbach2024deanonymizing}. Second, the security implications of slashing non-attributable misbehavior were analyzed in~\cite{badslashing}.

Consequently, this work does not disclose any new vulnerabilities. Instead, we assess the potential economic impact of previously known attacks and analyze the probability of their success under different adversarial capabilities. Our analysis is purely evaluative and does not introduce new attack techniques or lower the barrier to executing existing ones. Notably, the attacks discussed require an adversary to control assets valued in the tens of billions of dollars.

We carefully evaluated the ethical implications of this study and took steps to ensure responsible handling of sensitive information. In particular, we avoid operational details that could facilitate exploitation and focus on high-level modeling and analysis. Our goal is to inform protocol designers and the research community about systemic risks and to discuss mitigation strategies, rather than to enable real-world attacks.

As such, our research aims to protect existing stakeholders by preventing unjust penalization caused by protocol-level failures or attacks. Furthermore, our evaluation does not involve active experimentation on the Ethereum mainnet.

In detail:

\textbf{Stakeholders:} The primary stakeholders impacted by this research are users and validators on the Ethereum network, as well as companies and developers in the broader Ethereum ecosystem.

\textbf{Potential Harms:} Our analysis could inform adversaries about the feasibility of known attacks and could negatively impact the reputation of Ethereum, potentially translating into financial harm for users and validators.

\textbf{Mitigation:} We avoid providing operational details and focus on high-level analysis. Additionally, we propose protocol-level changes that address the identified risks and aim to improve the Ethereum ecosystem for all stakeholders.

\textbf{Ethical Principles and Decision:} As the outlined attacks are already known, we judge that proceeding with the research and its publication is ethically justified. The potential benefits for stakeholders and the importance of the long-term security of the Ethereum network outweigh the potential downsides.

\bibliographystyle{plainurl}
\bibliography{main}

\begin{thebibliography}{10}

\bibitem{randao}
Randao.
\newblock \url{https://github.com/randao/randao}.
\newblock Accessed on 26.08.2025.

\bibitem{aptos}
Aptos.
\newblock Run a validator and vfn.
\newblock \url{https://aptos.dev/network/nodes/validator-node}, 2025.
\newblock Accessed on 11.08.2025.

\bibitem{ffg}
Aditya Asgaonkar, Francesco D'Amato, Roberto Saltini, Luca Zanolini, and Chenyi Zhang.
\newblock A fast confirmation rule for the ethereum consensus protocol.
\newblock {\em arXiv preprint arXiv:2405.00549}, 2024.

\bibitem{cobra}
Zeta Avarikioti, Eleftherios~Kokoris Kogias, Ray Neiheiser, and Christos Stefo.
\newblock Cobra: A universal strategyproof confirmation protocol for quorum-based proof-of-stake blockchains, 2025.
\newblock URL: \url{https://arxiv.org/abs/2503.16783}, \href {https://arxiv.org/abs/2503.16783} {\path{arXiv:2503.16783}}.

\bibitem{beaconchain}
Beaconcha.in.
\newblock History of daily active validators.
\newblock \url{https://beaconcha.in/charts/validators}, 2025.
\newblock Accessed on 11.08.2025.

\bibitem{beaconchainentities}
Beaconcha.in.
\newblock Ethereum staking ecosystem overview.
\newblock \url{https://beaconcha.in/entities}, 2026.
\newblock Accessed on 11.08.2025.

\bibitem{bls2}
Dan Boneh, Manu Drijvers, and Gregory Neven.
\newblock Compact multi-signatures for smaller blockchains.
\newblock In Thomas Peyrin and Steven Galbraith, editors, {\em Advances in Cryptology -- ASIACRYPT 2018}, pages 435--464, Cham, 2018. Springer International Publishing.

\bibitem{bls}
Dan Boneh, Ben Lynn, and Hovav Shacham.
\newblock Short signatures from the weil pairing.
\newblock {\em Journal of cryptology}, 17(4):297--319, 2004.

\bibitem{10.1145/3670865.3673548}
Eric Budish, Andrew Lewis-Pye, and Tim Roughgarden.
\newblock The economic limits of permissionless consensus.
\newblock In {\em Proceedings of the 25th ACM Conference on Economics and Computation}, EC '24, page 704–731, New York, NY, USA, 2024. Association for Computing Machinery.
\newblock \href {https://doi.org/10.1145/3670865.3673548} {\path{doi:10.1145/3670865.3673548}}.

\bibitem{singleslotpath}
Vitalik Buterin.
\newblock Paths toward single-slot finality.
\newblock \url{https://notes.ethereum.org/@vbuterin/single_slot_finality}, 2022.
\newblock Accessed on 11.08.2025.

\bibitem{pbft}
Miguel Castro, Barbara Liskov, et~al.
\newblock Practical byzantine fault tolerance.
\newblock In {\em OsDI}, volume~99, pages 173--186, 1999.

\bibitem{algorand}
Jing Chen and Silvio Micali.
\newblock Algorand: A secure and efficient distributed ledger.
\newblock {\em Theoretical Computer Science}, 777:155--183, 2019.
\newblock In memory of Maurice Nivat, a founding father of Theoretical Computer Science - Part I.
\newblock URL: \url{https://www.sciencedirect.com/science/article/pii/S030439751930091X}, \href {https://doi.org/10.1016/j.tcs.2019.02.001} {\path{doi:10.1016/j.tcs.2019.02.001}}.

\bibitem{CohenP23}
Bram Cohen and Krzysztof Pietrzak.
\newblock Chia greenpaper, 2023.
\newblock URL: \url{https://docs.chia.net/green-paper-abstract}.

\bibitem{sff}
Francesco D’Amato and Luca Zanolini.
\newblock A simple single slot finality protocol for ethereum.
\newblock In {\em European Symposium on Research in Computer Security}, pages 376--393. Springer, 2023.

\bibitem{ethbls}
Ben Edgington.
\newblock Upgrading ethereum, bls signatures.
\newblock \url{https://eth2book.info/latest/part2/building_blocks/signatures/}, 2025.
\newblock Accessed on 30.07.2025.

\bibitem{ethereumbookcommittees}
Ben Edgington.
\newblock Upgrading ethereum, committees.
\newblock \url{https://eth2book.info/latest/part2/building_blocks/committees/}, 2025.
\newblock Accessed on 30.07.2025.

\bibitem{penalties}
Ethereum.
\newblock Sync committee penalties.
\newblock \url{https://eth2book.info/latest/part2/incentives/penalties/}, 2025.
\newblock Accessed on 21.08.2025.

\bibitem{feller1}
William Feller.
\newblock {\em An Introduction to Probability Theory and Its Applications}, volume~1.
\newblock Wiley, January 1968.
\newblock URL: \url{http://www.amazon.ca/exec/obidos/redirect?tag=citeulike04-20{\&}path=ASIN/0471257087}.

\bibitem{heimbach2024deanonymizing}
Lioba Heimbach, Yann Vonlanthen, Juan Villacis, Lucianna Kiffer, Roger Wattenhofer, et~al.
\newblock Deanonymizing ethereum validators: The p2p network has a privacy issue.
\newblock {\em arXiv preprint arXiv:2409.04366}, 2024.

\bibitem{DBLP:conf/podc/KeidarKNS21}
Idit Keidar, Eleftherios Kokoris{-}Kogias, Oded Naor, and Alexander Spiegelman.
\newblock All you need is {DAG}.
\newblock In {\em {PODC}}, pages 165--175. {ACM}, 2021.

\bibitem{byzcoin}
Eleftherios~Kokoris Kogias, Philipp Jovanovic, Nicolas Gailly, Ismail Khoffi, Linus Gasser, and Bryan Ford.
\newblock Enhancing bitcoin security and performance with strong consistency via collective signing.
\newblock In {\em 25th USENIX Security Symposium (USENIX Security 16)}, pages 279--296, Austin, TX, August 2016. USENIX Association.

\bibitem{omniledger}
Eleftherios Kokoris-Kogias, Philipp Jovanovic, Linus Gasser, Nicolas Gailly, Ewa Syta, and Bryan Ford.
\newblock Omniledger: A secure, scale-out, decentralized ledger via sharding.
\newblock In {\em 2018 IEEE Symposium on Security and Privacy (SP)}, pages 583--598, San Francisco, CA, USA, 2018. IEEE.

\bibitem{roaring}
Daniel Lemire, Gregory Ssi-Yan-Kai, and Owen Kaser.
\newblock Consistently faster and smaller compressed bitmaps with roaring.
\newblock {\em Software: Practice and Experience}, 46(11):1547--1569, 2016.

\bibitem{eclipse2}
Yuval Marcus, Ethan Heilman, and Sharon Goldberg.
\newblock Low-resource eclipse attacks on ethereum’s peer-to-peer network.
\newblock {\em IACR ePrint Cryptology Report}, 2020.

\bibitem{bitcoin}
Satoshi Nakamoto.
\newblock Bitcoin: A peer-to-peer electronic cash system.
\newblock 2008.

\bibitem{kauri}
Ray Neiheiser, Miguel Matos, and Lu\'{\i}s Rodrigues.
\newblock Kauri: Scalable bft consensus with pipelined tree-based dissemination and aggregation.
\newblock In {\em Proceedings of the ACM SIGOPS 28th Symposium on Operating Systems Principles}, SOSP '21, page 35–48, New York, NY, USA, 2021. Association for Computing Machinery.
\newblock \href {https://doi.org/10.1145/3477132.3483584} {\path{doi:10.1145/3477132.3483584}}.

\bibitem{largeststock}
World~Federation of~Exchanges.
\newblock Market statistics - february 2026.
\newblock \url{https://focus.world-exchanges.org/issue/february-2026/market-statistics}, 2026.
\newblock Accessed on 26.01.2026.

\bibitem{badslashing}
Ulysse Pavloff, Yackolley Amoussou-Guenou, and Sara Tucci-Piergiovanni.
\newblock Byzantine attacks exploiting penalties in ethereum pos.
\newblock In {\em 2024 54th Annual IEEE/IFIP International Conference on Dependable Systems and Networks (DSN)}, pages 53--65, 2024.
\newblock \href {https://doi.org/10.1109/DSN58291.2024.00020} {\path{doi:10.1109/DSN58291.2024.00020}}.

\bibitem{ethereum4sresearch}
Yiannis Psaras.
\newblock Gossipsub message propagation latency.
\newblock \url{https://ethresear.ch/t/gossipsub-message-propagation-latency}, 2024.
\newblock Accessed on 11.08.2025.

\bibitem{avalanche}
Team Rocket, Maofan Yin, Kevin Sekniqi, Robbert van Renesse, and Emin~Gün Sirer.
\newblock Scalable and probabilistic leaderless bft consensus through metastability, 2020.
\newblock URL: \url{https://arxiv.org/abs/1906.08936}, \href {https://arxiv.org/abs/1906.08936} {\path{arXiv:1906.08936}}.

\bibitem{netem}
HEMMINGER S.
\newblock Network emulation with netem.
\newblock https://cir.nii.ac.jp/crid/1572543024894323456, 2005.
\newblock Accessed on 18.04.2022.

\bibitem{figmentrewards}
Reza Sabernia.
\newblock Navigating rewards, risks, and attestation efficiency.
\newblock \url{https://figment.io/insights/strategies-for-ethereum-validators-navigating-rewards-risks-and-attestation-efficiency}, 2024.
\newblock Accessed on 30.07.2025.

\bibitem{gateways}
Paulo Silva, David Vavricka, João Barreto, and Miguel Matos.
\newblock Impact of geo-distribution and mining pools on blockchains: A study of ethereum.
\newblock In {\em 2020 50th Annual IEEE/IFIP International Conference on Dependable Systems and Networks (DSN)}, pages 245--252, 2020.
\newblock \href {https://doi.org/10.1109/DSN48063.2020.00041} {\path{doi:10.1109/DSN48063.2020.00041}}.

\bibitem{pectra}
Corwin Smith, Nicolas Consigny, Julio, nixo, Tim Beiko, Sam Calder-Mason, Mario Havel, and wackerow.
\newblock Pectra.
\newblock \url{https://ethereum.org/roadmap/pectra/}, 2025.
\newblock Accessed on 17.09.2025.

\bibitem{statista_crypto2025}
{Statista}.
\newblock Share of cryptocurrency owners in 53 countries and territories worldwide as of january 2025, 2025.
\newblock Accessed on 08.10.2025.
\newblock URL: \url{https://www.statista.com/forecasts/1452605/share-of-cryptocurrency-owners-in-selected-countries-worldwide}.

\bibitem{blst}
Supranational.
\newblock blst.
\newblock \url{https://github.com/supranational/blst}, 2025.
\newblock Accessed on 11.08.2025.

\bibitem{ethereumvalue}
tokenterminal.
\newblock Ecosystem total value locked.
\newblock \url{https://tokenterminal.com/explorer/projects/ethereum/ecosystem/ecosystem-tvl}, 2026.
\newblock Accessed on 26.01.2026.

\bibitem{ethereum}
Gavin Wood.
\newblock Ethereum: A secure decentralised generalised transaction ledger.
\newblock {\em Ethereum project yellow paper}, 151:1--32, 2014.
\newblock Accessed on 18.04.2022.
\newblock URL: \url{https://files.gitter.im/ethereum/yellowpaper/VIyt/Paper.pdf}.

\bibitem{eclipse1}
Guangquan Xu, Bingjiang Guo, Chunhua Su, Xi~Zheng, Kaitai Liang, Duncan~S. Wong, and Hao Wang.
\newblock Am i eclipsed? a smart detector of eclipse attacks for ethereum.
\newblock {\em Computers \& Security}, 88:101604, 2020.
\newblock URL: \url{https://www.sciencedirect.com/science/article/pii/S0167404818313798}, \href {https://doi.org/10.1016/j.cose.2019.101604} {\path{doi:10.1016/j.cose.2019.101604}}.

\bibitem{CTT}
Peiyun Zhang, Fuya Xu, Tianlin Huang, Haibin Zhu, and Qinglin Zhao.
\newblock Ctt: A three-layer tree consensus mechanism for consortium blockchains with enhanced security and reduced communication cost.
\newblock {\em IEEE Transactions on Industrial Informatics}, 21(6):4355--4366, 2025.
\newblock \href {https://doi.org/10.1109/TII.2025.3534426} {\path{doi:10.1109/TII.2025.3534426}}.

\end{thebibliography}

\section{Omitted Proofs and Analysis}\label{app:proofs}

\subsection{Proof of~\Cref{lem:internal}}\label{app:proof_internal}
\internal*

\begin{proof}
We first compute the probability of consistently picking a correct aggregate.
At the first internal node level (i.e., the internal nodes at depth $d-1$), each internal node receives a total of $16$ aggregates over the same set of leaf nodes. Of which $\frac{16}{3}$ are faulty and $16-\frac{16}{3}$ are correct. 
After removing the largest aggregate, we have 15 aggregates left.
As such, the probability to randomly choose a correct aggregate out of the 15 remaining aggregates is $\frac{16-\frac{16}{3}}{15}$.

At the subsequent internal node levels (i.e., at depths $d-2$ until depth 1), the correct nodes will, in the worst case, already have chosen the largest aggregate from a faulty aggregator, and as such, the probability to randomly pick an aggregate from a correct node is $\frac{2}{3}$.
Thus, at depth $d-2$ from the first internal aggregator to the proposer, the probability to consistently pick a correct aggregate by random over all $d-2$ depths is at least $\frac{2}{3}^{d-2}\cdot \frac{16-\frac{16}{3}}{15}$. 

Now we compute the probability that the proposer is correct.
As there is only a single proposer, there is only a  $\frac{2}{3}$ chance for the proposer to be correct. 
As such, the final probability for a leaf node to successfully be included in a single slot in the random aggregate is at least $\frac{2}{3}\cdot \frac{2}{3}^{d-2} \cdot \frac{16-\frac{16}{3}}{15}$. 
\end{proof}

\subsection{Analysis of \sys under slowly adaptive adversaries}\label{app:proof_adaptive}
In order to analyse the vote-censorship resilience of \sys under slowly adaptive adversaries, we first define a two-player, one-round, strategic game $\Gamma$ played by the adversary and the environment. 
Let us define a topological ordering of the nodes in \sys over all slots in the epoch, which we denote $T$. 
Note that we can always do so since \sys is directed and acyclic, and the slots are strictly sequential.
As the adversary can corrupt up to $f$ validators per epoch, the action space for the adversary is the set of all subsets of size $\leq f$ chosen from $[N]$.
The action space for the environment is a permutation $\pi: [N] \rightarrow [N]$, where $\{\pi(1), \dots, \pi(N)\}$ forms the resultant topological ordering or placement of the validators in the trees defined by \sys in the epoch.
The utility function of the adversary is $u_\mathcal{A}(a)$ given some action $a$ played by the adversary is $\alpha \cdot n_c$, where $n_c$ is the number of honest validators that the adversary managed to censor at the end of the epoch after playing action $a$ and $\alpha>0$ is some constant that represents the proportional increase in adversarial stake (over the honest) following a successful attack.
We denote the expected utility of the adversary when sampling an action from some distribution $\Delta$ by $\mathbb{E}_{a \leftarrow \Delta}[u_\mathcal{A}(a)]$.
The adversary and the environment both move simultaneously, and the move of the environment is always a random permutation over $[N]$ (i.e., each one of the $N!$ permutations over $[N]$ are selected with equal probability).
Let $\sigma$ be the adversarial strategy that chooses $f$ nodes uniformly at random from $N$ to corrupt. 
Formally, let $X$ denote the set of all subsets of size $f$ that can be drawn from $[N]$.
The strategy $\sigma$ samples a subset $x \leftarrow X$ uniformly at random.
We now show that $\sigma$ is optimal, i.e., no other strategy can increase the expected utility of the adversary. 

\begin{restatable}[]{lemma}{adap}\label{thm:adaptive}
Assuming the permutation over $[N]$ is perfectly random, $\sigma$ is optimal.
\end{restatable}

\begin{proof}
    We first show that always selecting $f$ validators to corrupt is optimal.
    Suppose we have a strategy $\sigma'$ with support over $X \cup a$ for some $a$ which is a subset of $[N]$ of size less than $f$.
    Let us define the set $A$ to be the set of subsets of $[N]$ of size $f$ such that each element of $A$ contains $a$.
    We now construct a second strategy $\sigma''$ that shifts the probability mass of $\sigma'$ on $a$ to some random element $a' \in A$.
    Since $n_c$ is monotone increasing in $f$, the utility of the adversary when playing $a'$ is never less than the utility of the adversary when playing $a$.
    Thus, the expected utility of the adversary under $\sigma''$ is never less than $\sigma'$. 

    Now we show by contradiction that there is no other strategy $\sigma'$ over $X$ that gives larger expected utility compared to $\sigma$.
    Suppose there is some strategy $\sigma'$ such that $\mathbb{E}_{a\leftarrow \sigma'}[u_{\mathcal{A}}(a)]  > \mathbb{E}_{a\leftarrow \sigma}[u_{\mathcal{A}}(a)]$. 
    This means there is some action $a$ in the supports of both $\sigma'$ and $\sigma$ such that given $a$, $u_{\mathcal{A}}(a) > u_{\mathcal{A}}(a)$. 
    Denote by $\beta >0$ the probability that this happens. 
    We can use $\sigma'$ to then construct a second adversary that breaks the perfect randomness of the permutation over $[N]$ with probability $\frac{\beta}{poly(\lambda)}$, where $\lambda$ is the security parameter of the original adversary.
    This contradicts the perfect randomness assumption. 
\end{proof}

\end{document}